%% file: main.tex
\documentclass[11pt,a4paper]{article}

\input{preamble}
\input{macros}
\title{A Deterministic $0.401$-Approximation for Non-monotone
Submodular Maximization under Matroid Constraints}
\author[1,2]{Shengminjie Chen} 
\author[1,2]{Kaifeng Lin}
\author[1,2]{Xiaoming Sun}
\author[1,2]{Jialin Zhang}

\affil[1]{State Key Lab of Processors,
	Institute of Computing Technology,
	Chinese Academy of Sciences,
	Beijing 100190, China}

\affil[2]{School of Computer Science and Technology,
	University of Chinese Academy of Sciences,
	Beijing 100049, China}

\begin{document}
\maketitle
\begin{abstract}
Submodular maximization arises in a wide range of combinatorial optimization problems and has been extensively studied from the perspective of approximation algorithms. For maximizing nonnegative, non-monotone submodular functions subject to a matroid constraint, randomized algorithms achieve a $0.401$-approximation, whereas obtaining comparable deterministic guarantees remains challenging. In this paper, we present a deterministic $0.401$-approximation algorithm in the exact value and independence oracle model, improving upon the previous deterministic guarantee of $(0.385-\varepsilon)$ \cite{ChenGLSZ2026}. 

Our algorithm builds on the extended multilinear extension and provides a deterministic implementation of the recursive guided framework underlying the randomized $0.401$-approximation algorithm. The main technical ingredient is a direction oracle supported on at most two partitioned matroid bases. This oracle simultaneously satisfies the growth condition required by the extended-multilinear update and an approximate stationarity condition needed to construct the recursive guides. The resulting finite-support representation enables exact evaluation of the relevant expectations and marginal values without random sampling. The algorithm uses polynomially many value- and independence-oracle queries.
\end{abstract}

\input{sections/01-introduction}
\input{sections/02-preliminaries}
\input{sections/03-algorithm}
\input{sections/04-analysis}
\input{sections/04-complexity}
\input{sections/05-conclusion}

\bibliographystyle{plain}
\bibliography{references}

\end{document}

%% file: preamble.tex
\usepackage[T1]{fontenc}
\usepackage{natbib}
\usepackage{authblk}
\usepackage{lmodern}
\usepackage[margin=1in]{geometry}
\usepackage{amsmath,amssymb,amsthm}
\usepackage{microtype,needspace,etoolbox}
\usepackage{enumitem}
\usepackage[ruled,vlined,linesnumbered]{algorithm2e}
\usepackage{flafter}
\usepackage[
  hidelinks,
  pdftitle={A Deterministic 0.401-Approximation Algorithm for Non-monotone Submodular Maximization under a Matroid Constraint},
  pdfauthor={Anonymous manuscript},
  pdfsubject={Deterministic approximation algorithms for submodular maximization}
]{hyperref}
\usepackage{cleveref}

\allowdisplaybreaks
\setlist{nosep}
\newtheorem{theorem}{Theorem}
\newtheorem{lemma}[theorem]{Lemma}
\newtheorem{corollary}[theorem]{Corollary}

\theoremstyle{definition}

\BeforeBeginEnvironment{lemma}{\Needspace{5\baselineskip}}
\BeforeBeginEnvironment{theorem}{\Needspace{5\baselineskip}}
\BeforeBeginEnvironment{corollary}{\Needspace{4\baselineskip}}

\SetKw{Return}{return}
\SetKw{Continue}{continue}
\SetKw{Break}{break}
\SetKwInput{KwIn}{Input}
\SetKwInput{KwOut}{Output}
\SetAlgoCaptionSeparator{.}
\SetAlCapFnt{\small\bfseries}
\SetAlCapNameFnt{\small\bfseries}
\SetAlFnt{\small}
\DontPrintSemicolon

%% file: macros.tex
\newcommand{\E}{\mathbb E}

\newcommand{\OPT}{\operatorname{OPT}}

\newcommand{\calZ}{\mathcal Z}
\newcommand{\calT}{\mathcal T}

\DeclareMathOperator{\Mar}{Mar}

\DeclareMathOperator{\supp}{supp}
\DeclareMathOperator{\fracnum}{frac}

\newcommand{\R}{\mathbb{R}}
\newcommand{\Prb}{\mathbb{P}}
\newcommand{\1}{\mathbf{1}}
\newcommand{\calI}{\mathcal{I}}
\newcommand{\calM}{\mathcal{M}}
\newcommand{\rank}{\operatorname{rank}}

%% file: sections/01-introduction.tex
\section{Introduction}
\label{sec:introduction}
Submodularity, which formalizes the principle of diminishing returns, naturally arises in a broad range of combinatorial optimization problems, including maximum cut, facility location, and maximum entropy sampling~\cite{LeeMNS2010}. Consequently, the maximization of submodular functions subject to various combinatorial constraints has become a fundamental problem in operations research and theoretical computer science. Since submodular maximization is $\mathcal{NP}$-hard in general, approximation algorithms have been a central focus.

The seminal work of Nemhauser, Wolsey, and Fisher~\cite{NemhauserWF1978} established the $(1-1/e)$-approximation guarantee of the greedy algorithm for monotone submodular maximization under a cardinality constraint. This optimal approximation ratio was subsequently extended to a knapsack constraint by Sviridenko~\cite{Sviridenko2004} using partial enumeration, and to a matroid constraint by C{\u a}linescu et al.~\cite{CalinescuCPV2011} through the continuous greedy framework. For the matroid constraint, Filmus and Ward~\cite{filmus2014} later developed a combinatorial approach based on non-oblivious local search, providing an alternative to continuous-relaxation methods. More recently, Buchbinder and Feldman~\cite{BuchbinderF2024Deterministic} obtained a deterministic $(1-1/e-\varepsilon)$-approximation, matching the best randomized guarantee up to an arbitrarily small loss.

For non-monotone submodular objectives, adding an element to a set may decrease its function value, posing additional challenges for approximation algorithm design. Feldman, Naor, and Schwartz~\cite{feldman2011} developed the measured continuous greedy algorithm, which achieves a $1/e$-approximation by controlling the coordinate-wise growth of a fractional solution in the multilinear relaxation. Ene and Nguyen~\cite{ene2016} subsequently improved the approximation guarantee to $0.372$ by restricting the growth of fractional coordinates and combining the continuous greedy approach with unconstrained submodular maximization. Buchbinder and Feldman~\cite{buchbinder2019} further improved the approximation ratio to $0.385$ through a non-symmetric technique that uses a fractional local-search solution to guide the measured continuous greedy algorithm. Recently, Buchbinder and Feldman~\cite{BuchbinderF2024DR} achieved a $0.401$-approximation by establishing sharper bounds for DR-submodular functions and developing a recursive framework based on Frank--Wolfe-guided measured continuous greedy.\footnote{Very recently, Aggarwal~\cite{Aggarwal2026Geometry} obtained an offline $(4/9-\varepsilon)$-approximation for maximizing a nonnegative, non-monotone DR-submodular function over a down-closed convex set. This result is established in a continuous optimization setting with first-order and projection oracles, whereas we consider a discrete set-function model with value and matroid-independence oracles.}

While randomized algorithms have progressively improved the approximation guarantees for non-monotone submodular maximization, the design of deterministic algorithms remains of particular interest, both from a theoretical perspective and for applications where reproducibility, worst-case guarantees, or the avoidance of random sampling is desirable. A central question in submodular maximization is whether deterministic algorithms can match the approximation guarantees obtained using randomization. This has motivated a parallel line of work seeking deterministic approximation algorithms that retain competitive approximation guarantees while maintaining low computational and oracle complexity.

For unconstrained maximization, Feige, Mirrokni, and Vondr{\'a}k~\cite{FeigeMV2011} gave a deterministic local-search algorithm with an approximation guarantee of $1/3-\varepsilon$. Buchbinder and Feldman~\cite{BuchbinderF2018} subsequently obtained the optimal deterministic $1/2$-approximation by derandomizing the double-greedy algorithm of Buchbinder et al.~\cite{BuchbinderFNS2015}. Their method explicitly maintains a distribution over algorithmic states and uses extreme-point solutions of linear programs to control its support. The inequalities needed by the approximation analysis are preserved in expectation over this distribution, rather than imposed on every state individually. Buchbinder and Feldman~\cite{BuchbinderF2018} also showed that the same distributional method yields a deterministic $1/e$-approximation under a cardinality constraint. Chen et al.~\cite{ChenNPK2024} later developed guided combinatorial algorithms with deterministic implementations that achieve $0.385-\varepsilon$ under cardinality constraint. They also obtained a deterministic $(0.377-\varepsilon)$-approximation in nearly linear time. For a single knapsack constraint, unequal element costs introduce an additional difficulty. Gupta et al.~\cite{GuptaRST2010} obtained an early deterministic $1/6$-approximation, and Sun et al.~\cite{SunZZZ2024} improved the ratio to $1/4$ using a twin-greedy approach. More recently, Chen et al.~\cite{ChenGLSZ2026} obtained a deterministic $(1/e-O(\varepsilon))$-approximation. A related line of work, including that of~\cite{LiuZ2026,PhamTHT2023,Pham2025}, focuses on reducing value-oracle complexity.

For a general matroid constraint, the deterministic approximation ratio improved through a sequence of combinatorial developments. The local-search framework of Lee et al.~\cite{LeeMNS2010} gives a $(1/4-\varepsilon)$-approximation for a single matroid. Han et al.~\cite{HanCCW2020} introduced Twin Greedy, which constructs two disjoint candidate solutions and achieves a $1/4$-approximation in $O(nr)$ time, where $n$ is the ground-set size and $r$ is the matroid rank. Their threshold-based variant achieves $1/4-\varepsilon$ in nearly linear time. Sun et al.~\cite{SunZZZ2024} improved the deterministic ratio to $0.283-o(1)$, and Chen et al.~\cite{ChenNPK2024} subsequently obtained $0.305-\varepsilon$ through guided combinatorial algorithms.

Further progress came from derandomizing the continuous relaxation. Algorithms based on the standard multilinear extension typically use sampling because exact evaluation can require exponentially many value-oracle queries. Buchbinder and Feldman~\cite{BuchbinderF2025Extension} introduced the extended multilinear extension (EME), which represents unions of independently activated sets. By maintaining a bounded number of nonzero set coordinates and rounding deterministically, they obtained a deterministic $(1/e-\varepsilon)$-approximation under a matroid constraint. Chen et al.~\cite{ChenGLSZ2026} combined this framework with a discrete local-search guide to obtain $0.385-\varepsilon$. These developments motivate studying how the EME can support 
the recursive guide construction of Buchbinder and
Feldman~\cite{BuchbinderF2024DR} without sampling.


\begin{table}[t]
\centering
\small
\renewcommand{\arraystretch}{1.12}
\caption{Deterministic approximation guarantees for non-negative
non-monotone submodular maximization under a general matroid constraint.}
\label{tab:deterministic-matroid-results}
\begin{tabular}{p{0.28\linewidth} c p{0.43\linewidth}}
\hline
Reference & Approximation & Main approach \\
\hline
Lee et al.~\cite{LeeMNS2010}
    & $1/4-\varepsilon$
    & Local search \\

Han et al.~\cite{HanCCW2020}
    & $1/4$
    & Twin greedy \\

Sun et al.~\cite{SunZZZ2024}
    & $0.283-o(1)$
    & Derandomization of Random Greedy \\

Chen et al.~\cite{ChenNPK2024}
    & $0.305-\varepsilon$
    & Guided combinatorial algorithm \\

Buchbinder and Feldman~\cite{BuchbinderF2025Extension}
    & $1/e-\varepsilon$
    & Extended multilinear extension \\

Chen et al.~\cite{ChenGLSZ2026}
    & $0.385-\varepsilon$
    & EME with a local-search guide \\

This work (Theorem \ref{thm:main-approximation})
    & $0.401-\varepsilon$
    & Recursive guided EME \\
\hline
\end{tabular}
\end{table}


Let $\mathcal N$ be a finite ground set, let $f:2^{\mathcal N}\to\mathbb R_{\ge 0}$ be a non-monotone submodular function, and let $\mathcal M=(\mathcal N,\mathcal I)$ be a matroid with independent set family $\mathcal I$. We consider
$$
  \max\{f(S):S\in\mathcal I\}.
$$
We write $\mathrm{OPT}$ for an optimal independent set. We assume access to an exact value oracle for $f$ and an independence oracle for $\mathcal M$. The objective need not be normalized; in particular, $f(\varnothing)$ may be positive. In this model, there is a natural question:




\textit{Can the recursive guided framework used in the randomized
$0.401$-approximation algorithm be implemented deterministically?}

Although the EME provides a possible way to derandomize the $0.401$-approximation algorithm~\cite{BuchbinderF2024DR}, implementing its recursive framework is not a direct application of the EME: derandomizing each
call must preserve both its growth guarantee and the guide certificate used by
the recursive analysis.  Our main technical contribution is a finite-support
direction oracle that preserves both properties without sampling.

\paragraph{Our contribution.} We answer this question affirmatively by giving a deterministic
$0.401$-approximation algorithm in the exact value- and
independence-oracle model. The formal statement is as follows, and the comparison with existing results is shown in Table \ref{tab:deterministic-matroid-results}.

\begin{theorem}
\label{thm:main-approximation}
Given a non-negative submodular
function through a value oracle and a matroid through an independence oracle, there is a deterministic algorithm that 
returns an independent set $S$ satisfying
\[
    f(S)\ge0.401f(\OPT).
\]
The algorithm uses polynomially many value- and independence-oracle queries.
\end{theorem}


Our algorithm retains the two complementary candidates in the randomized framework of Buchbinder and Feldman~\cite{BuchbinderF2024DR}: a box candidate, obtained by a weighted deterministic double-greedy procedure, and a candidate obtained by guided EME growth. Their guarantees apply in complementary regimes and must hold for the same guide distribution. The new direction oracle described below provides the common representation needed by both analyses.

The main technical contribution is a deterministic direction oracle supported on at most two \emph{partitioned bases}. A partitioned base is a tuple of disjoint blocks whose union is a matroid base, after augmentation with dummy elements. The blocks specify an EME update, while their union can serve as a recursive guide. The oracle uses a weighted local search, inspired by the deterministic matroid algorithm~\cite{BuchbinderF2024Deterministic}, whose objective combines the value of the underlying base with the aggregate block gain. Varying the weight and mixing at most two locally optimal outputs yields the required growth bound and approximate stationarity inequality whenever the target value satisfies the condition used in the analysis. Stationarity holds for the induced guide distribution; neither support base is required to satisfy it individually.

To implement guided growth, we explicitly maintain states consisting of a guide base and a list of EME blocks with their activation probabilities. For fixed accuracy, the number of states and the length of each list are bounded independently of the ground-set size. Every expectation and marginal value can therefore be evaluated as a finite sum. Deterministic EME rounding~\cite{BuchbinderF2025Extension} converts the terminal states into independent sets without decreasing their certified values. The box procedure uses the derandomized double-greedy~\cite{BuchbinderF2018}.

The recursive analysis uses the potential
$$
  \mathbb E_{Z\sim\mathcal Z}[f(\mathrm{OPT}\cup Z)]
$$
for an explicit guide distribution $\mathcal Z$. Starting from an approximately stationary guide, each call either produces the required growth candidate or generates an approximately stationary child guide whose potential decreases by at least $\varepsilon f(\mathrm{OPT})$. The algorithm explores all generated guides and need not identify the child certified by the analysis. Non-negativity and the initial upper bound of $2f(\mathrm{OPT})$ imply that a successful call is reached within bounded depth. At that call, the box and growth bounds hold for the same guide distribution and combine to yield $0.401$. Thus, the new direction oracle allows the recursive analysis in~\cite{BuchbinderF2024DR} to be implemented within the EME framework without sampling.

\paragraph{Organization.}
Section~\ref{sec:preliminaries} fixes the model and parameters.
Section~\ref{sec:algorithm} presents the deterministic subroutines and
their guarantees. Section~\ref{sec:analysis} proves the approximation
guarantee and analyzes the deterministic implementation and oracle
complexity. Section~\ref{sec:conclusion} concludes the paper.

%% file: sections/02-preliminaries.tex
\section{Preliminaries}
\label{sec:preliminaries}

This section fixes the notation needed by the algorithm and its analysis.
The EME framework and the deterministic rounding
theorem used below are due to Buchbinder and Feldman~\cite{BuchbinderF2025Extension}.

\subsection{Notation and submodular functions}

Let $\mathcal N_0$ be the finite input ground set of size $n$.  To index blocks and
rounds, for a positive integer $k$ let $[k]=\{1,\ldots,k\}$, and set
$[0]=\varnothing$.  To describe one-element insertions and deletions in
exchange arguments, for $S\subseteq\mathcal N_0$ and $u\in\mathcal N_0$ write
\[
  S+u=S\cup\{u\}
  \qquad\text{and}\qquad
  S-u=S\setminus\{u\}.
\]
The analysis measures the change caused by such updates through marginal
values.  For $f:2^{\mathcal N_0}\to\R$ and $A,S\subseteq\mathcal N_0$, define
\[
  f(u\mid S)=f(S+u)-f(S)
  \qquad\text{and}\qquad
  f(A\mid S)=f(S\cup A)-f(S).
\]
In particular, $f(u\mid S)=0$ whenever $u\in S$.

A set function $f$ is \emph{submodular} if
\[
  f(u\mid A)\ge f(u\mid B)
  \qquad
  \text{for all }A\subseteq B\subseteq\mathcal N_0
  \text{ and }u\in\mathcal N_0\setminus B.
\]
Equivalently,
\[
  f(A)+f(B)\ge f(A\cup B)+f(A\cap B)
  \qquad\text{for all }A,B\subseteq\mathcal N_0.
\]
Unless stated otherwise, the objective function is non-negative and
submodular, but it need not be monotone or normalized. Thus, marginal values
may be negative and $f(\varnothing)$ need not be zero.

To combine independent activation probabilities, for $a,b\in[0,1]$ define
their \emph{probabilistic sum} by
\[
  a\oplus b=1-(1-a)(1-b).
\]
This operation has a direct Bernoulli interpretation.  If
$X\sim\operatorname{Bernoulli}(a)$ and
$Y\sim\operatorname{Bernoulli}(b)$ are independent, then
\[
  \max\{X,Y\}\sim\operatorname{Bernoulli}(a\oplus b),
\]
because $\Prb[\max\{X,Y\}=0]=(1-a)(1-b)$.
More generally, for a finite index set $J$ and numbers
$\{a_j\}_{j\in J}\subseteq[0,1]$, define
\[
  \bigoplus_{j\in J}a_j=1-\prod_{j\in J}(1-a_j),
\]
where the empty probabilistic sum is zero.  For vectors in a unit cube,
$\oplus$ is applied coordinatewise.  To pass between sets and matroid
polytopes, we use $\1_S\in\{0,1\}^{\mathcal N_0}$ for the characteristic vector
of $S$; vector inequalities such as $x\le y$ are coordinatewise.

All algorithms are considered in the standard oracle model. A value-oracle
query returns $f(S)$ for a specified set $S\subseteq\mathcal N_0$, and an
independence-oracle query determines whether a specified set is independent
in the input matroid.

\subsection{Matroids}

We first recall bases, then augment the input matroid with dummy elements so
that every feasible solution has base form, and finally introduce the rank
and polyhedral notation used in the exchange and feasibility arguments.  We
use standard matroid facts from~\cite[Chapter~1]{Oxley2011} and
standard matroid-polytope facts from~\cite[Chapter~40]{Schrijver2003}.

A \emph{matroid} is a pair $\calM=(\mathcal N,\calI)$, where
$\calI\subseteq2^{\mathcal N}$ satisfies the following axioms:
\begin{enumerate}[label=(\roman*)]
  \item $\varnothing\in\calI$;
  \item if $A\subseteq B$ and $B\in\calI$, then $A\in\calI$;
  \item if $A,B\in\calI$ and $|A|<|B|$, then there exists
  $u\in B\setminus A$ such that $A+u\in\calI$.
\end{enumerate}
The members of $\calI$ are the \emph{independent sets}. An inclusionwise
maximal independent set is a \emph{base}. All bases have the same cardinality,
called the \emph{rank} of $\calM$.

\paragraph{Dummy elements.}
We use dummy elements to let every independent set be treated as a base without
changing its objective value.  Write
$\calM_0=(\mathcal N_0,\mathcal I_0)$ for the original input matroid, let $r$
be its rank, and let $D_{\mathrm{dum}}$ be a set of $r$ new elements.  The augmented matroid
$\calM=(\mathcal N,\calI)$ has ground set $\mathcal N=\mathcal N_0\cup D_{\mathrm{dum}}$ and
independent sets
\[
  \mathcal I
  =\{S\subseteq\mathcal N_0\cup D_{\mathrm{dum}}:
       S\cap\mathcal N_0\in\mathcal I_0,\ |S|\le r\}.
\]
The set $D_{\mathrm{dum}}$ is a base, and every independent set of the original matroid can
be extended to a base using elements of $D_{\mathrm{dum}}$. We extend every set function by
ignoring dummy elements, so that $f(S)=f(S\cap\mathcal N_0)$, and use the same
convention for auxiliary functions. Hence adding or removing dummy elements
does not change any function value.  Throughout the rest of the manuscript,
$\calM_0$ and $\mathcal N_0$ always denote the original input objects, whereas
$\calM$ and $\mathcal N$ denote their augmented counterparts;
$n=|\mathcal N_0|$ remains the number of original elements.  Dummy elements
are removed from every reported solution.  From this point onward, the set,
marginal, and characteristic-vector notation introduced above is applied to
the augmented ground set $\mathcal N$; the extension convention makes every
such expression agree with its restriction to $\mathcal N_0$.

To state feasibility arguments for arbitrary subsets, we use the rank
operator
\[
  \rank_{\calM}(S)
  =\max\{|I|:I\subseteq S,\ I\in\calI\}.
\]
When the matroid is clear from context, we omit the subscript $\calM$.

The following standard matching form of base exchange, also known as
Brualdi's lemma~\cite{Brualdi1969}, will be used when several one-element
exchanges are summed in a proof.

\begin{lemma}[Brualdi~\cite{Brualdi1969}]
\label{lem:base-exchange-matching}
Let $A$ and $B$ be bases of a matroid. There exists a bijection
$h:A\to B$ such that $h(u)=u$ for every $u\in A\cap B$ and
$B-h(u)+u$ is a base for every $u\in A$.
\end{lemma}

The \emph{matroid independence polytope} is
\[
  P(\calM)=\operatorname{conv}\{\1_I:I\in\calI\}.
\]
It is downward closed: if $y\in P(\calM)$ and $0\le x\le y$, then
$x\in P(\calM)$. The \emph{base polytope} is
\[
  B(\calM)=\operatorname{conv}\{\1_B:B\text{ is a base of }\calM\}.
\]
Equivalently, $B(\calM)$ is the face of $P(\calM)$ consisting of the vectors
$x$ satisfying $x(\mathcal N)=\rank(\mathcal N)$, where
$x(S)=\sum_{u\in S}x_u$.

Throughout the analysis, $\OPT$ denotes a maximum-value base of the augmented
matroid.  Such a base exists because any optimal independent set of
$\calM_0$ can be completed with dummy elements without changing its value.
Thus $f(\OPT)$ is also the optimum value of the original instance.

\subsection{The extended multilinear extension}

The EME encodes unions of independently activated sets.  We first define it
probabilistically, then give the sparse list representation used by the
algorithm, introduce the induced elementwise marginals, and finally invoke
deterministic rounding.

The extended multilinear extension (EME) associates one coordinate with every
subset of the ground set. For $y\in[0,1]^{2^{\mathcal N}}$ and
$A\subseteq\mathcal N$, let $R(y,A)$ be equal to $A$ with probability $y_A$
and to $\varnothing$ otherwise, independently for all $A$. Define
\[
  R(y)=\bigcup_{A\subseteq\mathcal N}R(y,A).
\]
The \emph{extended multilinear extension} of $f$ is the function
$F:[0,1]^{2^{\mathcal N}}\to\R$ given by
\begin{equation}
  F(y)=\E[f(R(y))].
  \label{eq:eme-probabilistic-definition}
\end{equation}
Equivalently,
\begin{equation}
  F(y)=
  \sum_{\mathcal A\subseteq2^{\mathcal N}}
  f\left(\bigcup_{A\in\mathcal A}A\right)
  \prod_{A\in\mathcal A}y_A
  \prod_{A\in2^{\mathcal N}\setminus\mathcal A}(1-y_A).
  \label{eq:eme-algebraic-definition}
\end{equation}
For $A\subseteq\mathcal N$, let $e_A$ denote the standard basis vector whose
$A$-coordinate is one. In particular, $F(e_A)=f(A)$.

To express the union of independently generated EME sets, for EME vectors
$y,z\in[0,1]^{2^{\mathcal N}}$ define
\[
  (y\vee z)_A=\max\{y_A,z_A\},
  \qquad
  (y\oplus z)_A=y_A\oplus z_A
  \quad(A\subseteq\mathcal N).
\]
If all activation trials used to construct $R(y)$ and $R(z)$ are mutually
independent, then $R(y\oplus z)$ and $R(y)\cup R(z)$ have the same
distribution.

\paragraph{List representation of EME}
The algorithms need not materialize all $2^{|\mathcal N|}$ coordinates of an
EME vector.  A list
\[
  \mathcal L=((A_1,p_1),\ldots,(A_m,p_m)),
  \qquad A_j\subseteq\mathcal N,\quad p_j\in[0,1],
\]
represents the experiment in which every listed set $A_j$ is activated
independently with probability $p_j$, and the activated sets are then united.
The EME vector induced by $\mathcal L$ is
\begin{equation}
  y_A^{\mathcal L}
  =\bigoplus_{j:A_j=A}p_j
  =1-\prod_{j:A_j=A}(1-p_j)
  \qquad(A\subseteq\mathcal N).
  \label{eq:eme-list-to-vector}
\end{equation}
Thus repeated occurrences of the same set may be merged into a single entry
without changing the distribution of the realized union. Conversely, every
EME vector $y$ has the canonical list representation
\[
  \mathcal L(y)=((A,y_A):A\in\supp(y)),
  \qquad
  \supp(y)=\{A\subseteq\mathcal N:y_A>0\}.
\]
Concatenating two such lists corresponds to taking the coordinatewise
probabilistic sum of their induced EME vectors. The list representation is
exact and sparse rather than an approximation of the EME vector.

The running time of EME rounding depends on the number of fractional
coordinates.  We denote this number by
\[
  \fracnum(y)
  =|\{A\subseteq\mathcal N:0<y_A<1\}|.
\]
When $y$ is represented by a list with $m$ entries, $F(y)$ can be
evaluated exactly using at most $2^m$ value-oracle queries by enumerating all
activation patterns. This bound counts list entries before repeated sets are
merged; using the canonical representation gives the corresponding bound in
terms of $|\supp(y)|$.

\paragraph{Marginal vector.}
Matroid feasibility of an EME vector is expressed through its elementwise
inclusion probabilities.  Accordingly, the marginal vector of $y$ is
$\operatorname{Mar}(y)\in[0,1]^{\mathcal N}$, where
\begin{equation}
  \operatorname{Mar}_u(y)
  =\Prb[u\in R(y)]
  =\bigoplus_{A\subseteq\mathcal N:\,u\in A}y_A
  =1-\prod_{A\subseteq\mathcal N:\,u\in A}(1-y_A).
  \label{eq:eme-marginal-vector}
\end{equation}
For a list $\mathcal L=((A_1,p_1),\ldots,(A_m,p_m))$, the same
quantity is computed directly as
\[
  \operatorname{Mar}_u(y^{\mathcal L})
  =\bigoplus_{j:u\in A_j}p_j
  =1-\prod_{j:u\in A_j}(1-p_j).
\]
In particular, feasibility of an EME solution under a matroid constraint will
be expressed through the condition $\operatorname{Mar}(y)\in P(\calM)$.

The deterministic EME rounding theorem of Buchbinder and
Feldman~\cite[Theorem~1.2]{BuchbinderF2025Extension} converts any such feasible EME solution into
an ordinary independent set without decreasing its objective value.

\begin{theorem}[Buchbinder and Feldman~\cite{BuchbinderF2025Extension}, Theorem~1.2]
\label{thm:eme-rounding}
Let $y$ be an EME vector satisfying
$\operatorname{Mar}(y)\in P(\calM)$.  There is a deterministic algorithm that
returns an independent set $S\in\calI$ such that
\[
  f(S)\ge F(y).
\]
It uses $O(n^2 2^{\fracnum(y)})$ value-oracle queries and
$O(n^5\log^2 n)$ independence-oracle queries.
\end{theorem}

We use the independence-polytope version of the theorem; a statement restricted
to the base polytope would not suffice for the EME vectors constructed later.
In particular, Theorem~\ref{thm:eme-rounding} implies
$F(y)\le f(\OPT)$ whenever
$\operatorname{Mar}(y)\in P(\calM)$.

\paragraph{A lower bound determined by marginal probabilities.}
To compare a random-set value using only its marginal vector, we use the
Lov\'asz extension.  For $\varphi:2^{\mathcal N}\to\R$, define
\[
  \widehat\varphi(x)
  =\int_0^1\varphi(\{u\in\mathcal N:x_u>t\})\,dt,
  \qquad x\in[0,1]^{\mathcal N}.
\]
Chen et al.~\cite[Lemma~10]{ChenGLSZ2026} prove the following inequality from the
Lov\'asz-extension representation.  We invoke their result directly rather
than repeat its proof; in particular, it does not require independence among
the elements of the random set.

\begin{lemma}[Chen et al.~\cite{ChenGLSZ2026}, Lemma~10]
\label{lem:lovasz-marginal-lower-bound}
Let $X\subseteq\mathcal N$ be a random set with marginal vector
$x_u=\Prb[u\in X]$. If $\varphi$ is submodular, then
\[
  \E[\varphi(X)]\ge\widehat\varphi(x).
\]
No non-negativity or independence assumption is required.
\end{lemma}

\subsection{Parameters}

We collect here the numerical parameters used by the algorithm; the
corresponding algorithmic objects are defined in
Section~\ref{sec:algorithm}.  Let $r=\rank_{\calM}(\mathcal N)$, and let
$S_{\max}$ maximize $f$ among $\varnothing$ and the feasible non-dummy
singletons.  We assume $r\ge1$, $f(S_{\max})>0$, and
$0<\varepsilon\le1/10$.

Set
\begin{equation}
  \ell=\left\lceil\frac{1}{\varepsilon}\right\rceil,
  \qquad
  \delta=\frac{1}{\left\lceil 2\ell^2/\varepsilon\right\rceil},
  \qquad
  t_s=0.3682,
  \qquad
  i_s=\left\lceil\frac{t_s}{\delta}\right\rceil.
  \label{eq:parameters}
\end{equation}
Here $\ell$ is the number of blocks in a direction, $\delta$ is both the
activation probability and the time-step size, and $i_s$ is the round at
which guided greedy changes phases.  These choices satisfy
\[
  \frac{1}{\ell}\le\varepsilon,
  \qquad
  2\ell^2\delta\le\varepsilon,
  \qquad
  0<t_s<i_s\delta<\frac{1}{2}.
\]

The local-search tolerance is
$\eta=\varepsilon f(S_{\max})/(4r)$, the recursion depth is at most
$\lceil2/\varepsilon\rceil+1$, and the box routine uses $\rho=2.22$.
Submodularity and non-negativity give
$f(\OPT)\le n f(S_{\max})$, which will be used to bound the additive errors.
The final approximation guarantee uses $\varepsilon=10^{-8}$, while
$O_{\varepsilon}(\cdot)$ suppresses factors depending only on $\varepsilon$.
All tie-breaking, exchange, and enumeration orders are fixed in advance.

%% file: sections/03-algorithm.tex
\section{Algorithm}
\label{sec:algorithm}

\subsection{Overview}

We begin by recalling the high-level structure of the randomized
$0.401$-approximation of Buchbinder and Feldman~\cite{BuchbinderF2024DR}.  In that
framework, a fractional solution $z$ in a down-closed relaxation serves as a
\emph{fractional guide vector}.  The algorithm uses $z$ in two complementary
ways.  Box maximization searches for a valuable solution dominated by $z$, whereas
Frank--Wolfe guided measured continuous greedy starts from the zero vector
and constructs a second solution through a sequence of feasible fractional
growth directions, using $z$ to guide the initial part of the process.

The guided greedy analysis has two possible outcomes.  If the selected
fractional growth directions provide sufficient progress throughout the
process, its final solution and
the box candidate together yield the desired approximation
guarantee.  Otherwise, one of the fractional growth directions is itself a
better fractional guide vector in the sense required by the recursive
analysis.  The algorithm cannot identify this vector because the criterion
defining such an improvement
depends on an unknown optimum.  It therefore recursively applies the same two
procedures to every generated fractional growth direction, treating each as a
new fractional guide vector.
Along an appropriate recursion path, every unsuccessful call improves a
bounded potential that measures the relation between the fractional guide
vector and the optimum by a fixed additive amount; hence, a successful call is eventually
reached.  Combining the complementary guarantees of box maximization and the
guided greedy procedure at that call gives the $0.401$ approximation.

\subsection{Main procedure}

The deterministic algorithm retains the recursive structure described above,
but it cannot manipulate an implicitly sampled fractional guide.  Moreover,
the analysis needs to evaluate the guide's intersection and union with other
sets and to run box maximization on the guide itself.  We therefore replace a
fractional guide by a finite weighted collection of bases.  The direction
oracle constructed later always produces a collection supported on at most
two bases, so this representation is also preserved when a direction becomes
the guide of a recursive call.

\paragraph{Guide distributions.}
A guide distribution is represented explicitly as
\[
    \calZ=\{(Z_1,\theta),(Z_2,1-\theta)\},
    \qquad \theta\in[0,1],
\]
where $Z_1$ and $Z_2$ are matroid bases.  A guide distribution supported on a
single base is included by allowing $\theta\in\{0,1\}$.  For any function $h$
defined on sets, we write the corresponding weighted average as
\[
    \E_{Z\sim\calZ}[h(Z)]
    =
    \theta h(Z_1)+(1-\theta)h(Z_2).
\]
The algorithm does not sample $Z$; it stores both bases and their weights and
evaluates this expression directly.

We call a guide distribution $\calZ$ \emph{$\varepsilon$-stationary} if
\begin{equation}
    \E_{Z\sim\calZ}
    \bigl[f(Z\cup\OPT)+f(Z\cap\OPT)-2f(Z)\bigr]
    \le 2\varepsilon f(\OPT).
    \label{main:guide-stationarity}
\end{equation}
This is an analytical property: it measures whether the weighted guide is
adequate for the box bound or whether the recursive analysis must replace it
by an improved guide.  It is distinct from $\eta$-local optimality of the
local-search subroutine, which means that no allowed local move has tested
gain greater than $\eta$.

With the guide representation fixed, we can state the main recursion using
three deterministic subroutines.  One constructs the initial guide
distribution, one returns a collection of feasible sets by box maximization,
and one returns both feasible sets and improved guide distributions through
guided greedy.  These subroutines are specified later.

\begin{algorithm}[tb]
\DontPrintSemicolon
\KwIn{A non-negative submodular function $f$, a matroid $\calM_0$, and an
accuracy parameter $0<\varepsilon\le1/10$.}
\KwOut{An independent set.}

Compute the rank $r$ and $S_{\max}$ as in the preliminaries\;
\If{$r=0$ or $f(S_{\max})=0$}{
    \Return{$\varnothing$}\;
}
Construct the augmented matroid $\calM$ from $\calM_0$ as in the preliminaries\;
$\calZ_0\gets\operatorname{InitialGuideDistribution}(f,\calM)$\;
$d_{\max}\gets\lceil2/\varepsilon\rceil+1$\;

\SetKwFunction{Explore}{Explore}
\SetKwProg{Proc}{Procedure}{:}{end}

\Proc{\Explore{$\calZ,d$}}{
    $\mathcal C_{\mathrm{box}}
      \gets\operatorname{DeterministicBoxCandidates}(\calZ,\rho)$\;

    $(\mathcal C_{\mathrm{greedy}},\mathcal G)
      \gets\operatorname{DeterministicGuidedGreedy}
      (f,\calM,\calZ)$\;

    $\mathcal C
      \gets\mathcal C_{\mathrm{box}}
      \cup\mathcal C_{\mathrm{greedy}}$\;

    \If{$d<d_{\max}$}{
        \ForEach{$\calZ'\in\mathcal G$}{
            $\mathcal C
              \gets\mathcal C
              \cup\{\Explore(\calZ',d+1)\}$\;
        }
    }

    \Return{a maximum-value set in $\mathcal C$}\;
}

$S\gets\Explore(\calZ_0,0)$\;
\Return{$S$ with dummy elements removed}\;

\caption{Main deterministic algorithm}
\label{alg:main-deterministic}
\end{algorithm}

\subsection{Deterministic box maximization}
\label{box:section}

The box maximization algorithm of Buchbinder and
Feldman~\cite{BuchbinderF2024DR} is derived from the double-greedy algorithm of
Buchbinder et al.~\cite{BuchbinderFNS2015}.  We implement it using a $\rho$-weighted
adaptation of the deterministic double-greedy construction of Buchbinder and
Feldman~\cite{BuchbinderF2018}.

Fix a base $Z$ in the support of the guide distribution.  The corresponding
box maximization problem is to find a
high-value set among all subsets of $Z$:
\[
    \max_{S\subseteq Z} f(S).
\]
The algorithm processes the elements of $Z$ in a fixed order.  As in
standard double-greedy, it maintains a lower set $X$ and an upper set
$Y$ satisfying
\[
    X\subseteq Y\subseteq Z.
\]
Elements in $X$ have been accepted, elements outside $Y$ have been
rejected, and the elements of $Y\setminus X$ remain undecided.  The
initial state is $(X,Y)=(\varnothing,Z)$.

When processing an element $u\in Y\setminus X$, define its addition
and deletion gains by
\[
    a=f(X+u)-f(X),
    \qquad
    b=f(Y-u)-f(Y).
\]
There are two possible updates:
\[
    (X,Y)\longrightarrow(X+u,Y)
    \qquad\text{and}\qquad
    (X,Y)\longrightarrow(X,Y-u).
\]
A randomized double-greedy algorithm chooses between these updates
randomly.  The deterministic construction instead maintains an
explicit weighted family of states and retains both updates whenever
they receive positive weight.

Formally, at the beginning of an iteration the algorithm stores
\[
    \mathcal D=\{(w_j,X_j,Y_j):j\in J\},
\]
where
\[
    w_j\ge 0,\qquad
    \sum_{j\in J}w_j=1,\qquad
    X_j\subseteq Y_j\subseteq Z.
\]
For the next element $u$, let
\[
    a_j=f(X_j+u)-f(X_j),
    \qquad
    b_j=f(Y_j-u)-f(Y_j).
\]
For every state, the algorithm assigns a number $p_j\in[0,1]$.
The state is then replaced by
\[
    (w_jp_j,X_j+u,Y_j)
 \qquad
 \text{and}
 \qquad    (w_j(1-p_j),X_j,Y_j-u),
\]
with zero-weight states discarded. The values $p_j$ are chosen to satisfy
\begin{align}
 \sum_{j\in J}w_j
 \left(
    \frac{p_ja_j}{2\rho}
    +\frac{\rho(1-p_j)b_j}{2}
 \right)
 &\ge
 \sum_{j\in J}w_jp_jb_j,
 \label{box:add-constraint}\\
 \sum_{j\in J}w_j
 \left(
    \frac{p_ja_j}{2\rho}
    +\frac{\rho(1-p_j)b_j}{2}
 \right)
 &\ge
 \sum_{j\in J}w_j(1-p_j)a_j.
 \label{box:delete-constraint}
\end{align}
Here, $\rho>0$ is the imbalance parameter fixed in the preliminaries.
These two inequalities ensure that the weighted progress of the
maintained states compensates for either possible change in a
comparison solution.

A feasible vector $(p_j)_{j\in J}$ can be found deterministically.
One first uses the statewise choices
\[
 p_j=
 \begin{cases}
  0,
      & a_j\le 0,\\
  1,
      & a_j>0\text{ and }b_j\le 0,\\
  \dfrac{a_j}{a_j+b_j},
      & a_j>0\text{ and }b_j>0.
 \end{cases}
\]
Submodularity gives $a_j+b_j\ge 0$.  The sign cases satisfy both constraints
directly.  When $a_j,b_j>0$, both right-hand sides contributed by state $j$
equal $a_jb_j/(a_j+b_j)$, while its contribution to the common left-hand
side is
\[
    \frac{a_j^2/\rho+\rho b_j^2}{2(a_j+b_j)}
    \ge\frac{a_jb_j}{a_j+b_j}.
\]
Thus these choices satisfy both aggregate inequalities.  Both constraint
slacks are affine functions of $p$.  If more than two coordinates are
fractional, the two-row coefficient matrix has a nonzero null direction
supported on those coordinates.  Moving along this direction preserves both
slacks, and one may continue until another coordinate reaches zero or one.
Repeating this operation
leaves at most two fractional coordinates.  Consequently, processing
one element increases the number of stored states by at most two, so
the state support remains polynomially bounded.

After every element of $Z$ has been processed, each surviving state
satisfies $X_j=Y_j$.  The algorithm returns a maximum-value set among
these final states.  We denote the returned set by
\[
    \operatorname{DeterministicBox}(Z,\rho).
\]
All element orders, nullspace choices, and tie-breaking rules are
fixed in advance, so the procedure is deterministic.

\begin{algorithm}[tb]
\DontPrintSemicolon
\KwIn{A non-negative submodular function $f$, a base $Z$, and an
imbalance parameter $\rho>0$.}
\KwOut{A set $S\subseteq Z$.}

$\mathcal D\gets\{(1,\varnothing,Z)\}$\;

\For{each $u\in Z$ in the fixed order}{
    Write $\mathcal D=\{(w_j,X_j,Y_j):j\in J\}$\;
    $a_j\gets f(X_j+u)-f(X_j)$ and
    $b_j\gets f(Y_j-u)-f(Y_j)$ for every $j\in J$\;

    \ForEach{$j\in J$}{
        $p_j\gets 0$ if $a_j\le 0$;
        $p_j\gets 1$ if $a_j>0$ and $b_j\le 0$;
        otherwise $p_j\gets a_j/(a_j+b_j)$\;
    }

    $I\gets\{j\in J:0<p_j<1\}$\;
    \While{$|I|>2$}{
        Choose a nonzero vector $d\in\mathbb R^J$, supported on $I$,
        such that replacing $p$ by $p+\lambda d$ preserves the slack in
        both \eqref{box:add-constraint} and
        \eqref{box:delete-constraint} for every $\lambda$\;
        $\lambda\gets\max\{t>0:p+td\in[0,1]^J\}$\;
        $p\gets p+\lambda d$ and
        $I\gets\{j\in J:0<p_j<1\}$\;
    }

    $\mathcal D\gets
      \{(w_jp_j,X_j+u,Y_j):j\in J,\ p_j>0\}
      \cup
      \{(w_j(1-p_j),X_j,Y_j-u):j\in J,\ p_j<1\}$\;
}

\Return{a maximum-value set among
$\{X_j:(w_j,X_j,Y_j)\in\mathcal D\}$}\;

\caption{$\operatorname{DeterministicBox}(Z,\rho)$}
\label{alg:deterministic-box}
\end{algorithm}

The guarantee is obtained by following a fixed comparison set through the
state distribution.  The two update constraints charge every loss in this
comparison value to increases in the lower- and upper-set potentials.  After
these charges are summed over the processed elements, the comparison set and
both potentials have the coefficients in the following lemma.

\begin{lemma}
\label{box:guarantee}
For every base $Z$, every set $Q\subseteq Z$, and every $\rho>0$, the
set $S=\operatorname{DeterministicBox}(Z,\rho)$ satisfies
\[
    f(S)
    \ge
    \frac{
        2\rho f(Q)+f(\varnothing)+\rho^2f(Z)
    }{(\rho+1)^2}.
\]
\end{lemma}

\begin{proof}
Fix $Q\subseteq Z$ for the analysis.  For every maintained state
$(X_j,Y_j)$, define
\[
    Q_j=(Q\cup X_j)\cap Y_j.
\]
Initially $Q_j=Q$, and after every element has been processed,
$Q_j=X_j=Y_j$.

Consider an iteration that processes $u$.  If $u\notin Q$, deleting $u$
from $Y_j$ leaves $Q_j$ unchanged, whereas adding $u$ to $X_j$ replaces
$Q_j$ by $Q_j+u$.  Since $Q_j\subseteq Y_j-u$, submodularity gives
\[
    f(Q_j)-f(Q_j+u)
    \le f(Y_j-u)-f(Y_j)
    =b_j.
\]
Hence, the aggregate decrease in the value of the comparison sets is at
most $\sum_jw_jp_jb_j$, which is bounded by the left-hand side of
\eqref{box:add-constraint}.

If $u\in Q$, adding $u$ to $X_j$ leaves $Q_j$ unchanged, whereas deleting
$u$ from $Y_j$ replaces $Q_j$ by $Q_j-u$.  Since
$X_j\subseteq Q_j-u$, submodularity gives
\[
    f(Q_j)-f(Q_j-u)
    \le f(X_j+u)-f(X_j)
    =a_j.
\]
The aggregate decrease is therefore at most
$\sum_jw_j(1-p_j)a_j$, which is bounded by the left-hand side of
\eqref{box:delete-constraint}.

Let
\[
    F_X=\sum_jw_jf(X_j),
    \qquad
    F_Y=\sum_jw_jf(Y_j),
    \qquad
    F_Q=\sum_jw_jf(Q_j)
\]
denote the three weighted values at the beginning of an iteration.  The
increase in $F_X$ during this iteration is
$\sum_jw_jp_ja_j$, and the increase in $F_Y$ is
$\sum_jw_j(1-p_j)b_j$.  Combining the preceding case analysis with
\eqref{box:add-constraint} and \eqref{box:delete-constraint} shows that
the decrease in $F_Q$ is at most
\[
    \frac{1}{2\rho}\,\Delta F_X
    +\frac{\rho}{2}\,\Delta F_Y.
\]

Initially,
\[
    (F_X,F_Y,F_Q)
    =\bigl(f(\varnothing),f(Z),f(Q)\bigr).
\]
At termination, $X_j=Y_j=Q_j$ in every state, so all three weighted
values are equal; denote their common value by $F_{\mathrm{out}}$.
Summing the per-iteration inequality therefore yields
\[
    f(Q)-F_{\mathrm{out}}
    \le
    \frac{F_{\mathrm{out}}-f(\varnothing)}{2\rho}
    +\frac{\rho}{2}
       \bigl(F_{\mathrm{out}}-f(Z)\bigr).
\]
Multiplying by $2\rho$ and rearranging gives
\[
    (\rho+1)^2F_{\mathrm{out}}
    \ge
    2\rho f(Q)+f(\varnothing)+\rho^2f(Z).
\]
The algorithm returns a maximum-value set in the final state support, whose
value is at least $F_{\mathrm{out}}$.  This proves the lemma.
\end{proof}

The comparison set $Q$ is used only in the analysis and is not given
to the algorithm.  Later, we apply the lemma with
$Q=\OPT\cap Z$.  For a guide distribution
\[
    \calZ=\{(Z_1,\theta),(Z_2,1-\theta)\},
\]
we run $\operatorname{DeterministicBox}$ independently on every positive-weight
base $Z_i$ and retain all returned sets as candidates.  We denote the resulting
family by $\operatorname{DeterministicBoxCandidates}(\calZ,\rho)$.  In
particular, this routine does not sample a base from $\calZ$.

\subsection{Deterministic direction oracle}
\label{dir:section}

The guided greedy procedure requires, in every round, a distribution over
partitioned bases whose blocks supply sufficient marginal progress and whose
induced distribution over underlying bases is an $\varepsilon$-stationary guide distribution
for a possible recursive call.  We call this output a \emph{direction
distribution}.  These requirements need not be met by a single partitioned
base, so the direction distribution may be supported on two partitioned
bases.  The underlying one-exchange local search originates in Buchbinder and
Feldman~\cite[Algorithm~3]{BuchbinderF2024Deterministic}; Chen et al.~\cite[Section~3.1]{ChenGLSZ2026}
adapt it to obtain a one-exchange local optimum for a non-monotone objective.
Our extended local search incorporates block reassignments and an auxiliary
weighted objective.  A binary search over the multiplier then produces two
$\eta$-locally optimal partitioned bases that can be mixed to attain a
prescribed aggregate block gain.

Let $g:2^{\mathcal N}\to\mathbb R_{\ge0}$ be a submodular auxiliary function.  The analysis below assumes
\begin{equation}
    g(\varnothing)\le f(\OPT),
    \qquad
    g(S)\le g(\varnothing)+f(\OPT)
    \quad\text{for every }S\in\calI.
    \label{dir:g-bounds}
\end{equation}
The guided greedy procedure will specify $g$ in each round.
Throughout this subsection, $f$ and $g$ are fixed within each call to the
direction oracle.  We therefore suppress them from the notation; the
multiplier $\lambda$ remains explicit because it varies during the search.

A \emph{partitioned base} is a tuple
$\calT=(T_1,\ldots,T_\ell)$ of pairwise disjoint sets whose union
\[
    T=q(\calT)=\bigcup_{b=1}^{\ell}T_b
\]
is a base of $\calM$; empty blocks are allowed.  Define its aggregate block
gain and its multiplier-dependent objective by
\begin{align}
    W(\calT)
    &=\sum_{b=1}^{\ell}
        \bigl(g(T_b)-g(\varnothing)\bigr),
    \label{dir:block-gain}\\
    H_\lambda(\calT)
    &=f(q(\calT))+\lambda W(\calT),
    \qquad \lambda\ge0.
    \label{dir:weighted-objective}
\end{align}

\subsubsection{Extended local search}

The local search uses exchange and reassignment moves.  Suppose that
$v\in T_k$, $u\notin T$, and $T-v+u$ is a base.  An exchange removes $v$
from its current block $T_k$ and assigns the entering element $u$ to a chosen
destination block $T_j$.  If $j\ne k$, the update replaces $T_k$ by $T_k-v$
and $T_j$ by $T_j+u$; if $j=k$, it replaces $T_k$ by $T_k-v+u$.  All other
blocks remain unchanged.  The tested gain of this exchange is
\begin{equation}
\begin{aligned}
    \Gamma_{\lambda,\calT}(u,v,j)
    &=f(u\mid T)-f(v\mid T-v)+
      \lambda\bigl(g(u\mid T_j)-g(v\mid T_k-v)\bigr).
\end{aligned}
\label{dir:exchange-gain}
\end{equation}
For $v\in T_k$ and $j\ne k$, a reassignment move transfers $v$ from
$T_k$ to $T_j$ without changing the underlying base.  Its tested gain is
\begin{equation}
    \Lambda_{\lambda,\calT}(v,j)
    =\lambda\bigl(g(v\mid T_j)-g(v\mid T_k-v)\bigr).
    \label{dir:reassignment-gain}
\end{equation}
The exchange quantity is a lower bound on the actual change in $H_\lambda$,
whereas the reassignment quantity is the exact change.
Both quantities depend on the current partitioned base and on the multiplier.
When $\lambda$ and $\calT$ are fixed by the current local-search state, we
write $\Gamma(u,v,j)$ and $\Lambda(v,j)$, respectively.

Set
\[
    \eta=\frac{\varepsilon f(S_{\max})}{4r}.
\]
A partitioned base is \emph{$\eta$-locally optimal} for $H_\lambda$ if
every feasible exchange satisfies $\Gamma(u,v,j)\le\eta$ and every
reassignment satisfies $\Lambda(v,j)\le\eta$.  Local optimality here is
defined with respect to these two discrete neighborhoods and their tested
gains.

\begin{algorithm}[tb]
\DontPrintSemicolon
\KwIn{Submodular functions $f$ and $g$, and a multiplier $\lambda\ge0$.}
\KwOut{An $\eta$-locally optimal partitioned base $\calT$.}

Initialize $\calT$ to a fixed partition of the all-dummy base\;

\While{a feasible exchange with $\Gamma(u,v,j)>\eta$ or a reassignment
with $\Lambda(v,j)>\eta$ exists}{
    Select the first such move in the fixed inspection order\;
    Apply the selected exchange or reassignment\;
}

\Return{$\calT$}\;

\caption{Extended Local Search}
\label{alg:extended-local-search}
\end{algorithm}

To justify this procedure, we first compare each tested exchange or
reassignment gain with the resulting change in $H_\lambda$.  An upper bound on
the range of $H_\lambda$ then limits the number of accepted moves and gives
the termination statement below.

\begin{lemma}
\label{lem:extended-local-search}
Every move performed by Algorithm~\ref{alg:extended-local-search} increases
$H_\lambda$ by more than $\eta$.  Moreover, the algorithm performs fewer than
\[
    \frac{(1+\lambda\ell)f(\OPT)}{\eta}
    \le
    \frac{4rn(1+\lambda\ell)}{\varepsilon}
\]
moves and terminates at an $\eta$-locally optimal partitioned base.  In particular, for
$0\le\lambda\le\lambda_{\max}$, the number of iterations is polynomial.
\end{lemma}

\begin{proof}
Fix a feasible exchange at the current partitioned base
$\calT=(T_1,\ldots,T_\ell)$.  Let $v\in T_k$ be the element removed from its
current block, let $u\notin T$ be the entering element, and let $T_j$ be its
chosen destination block.  The underlying base changes from $T$ to $T-v+u$,
so submodularity gives
\begin{align}
    f(T-v+u)-f(T)
    =f(u\mid T-v)-f(v\mid T-v)
    \ge f(u\mid T)-f(v\mid T-v).
    \label{dir:exchange-f-change}
\end{align}
Suppose first that $j\ne k$.  The operation updates
$T_k\leftarrow T_k-v$ and $T_j\leftarrow T_j+u$, while every other block is
unchanged.  Therefore
\begin{align*}
    \Delta W
    &=g(T_k-v)-g(T_k)+g(T_j+u)-g(T_j)=-g(v\mid T_k-v)+g(u\mid T_j).
\end{align*}
Combining this identity with \eqref{dir:exchange-f-change} yields
\[
    \Delta H_\lambda
    \ge\Gamma_{\lambda,\calT}(u,v,j).
\]

Suppose now that $j=k$.  The entering element is assigned to the same block
from which $v$ is removed, so the operation updates
$T_k\leftarrow T_k-v+u$ and leaves all other blocks unchanged.  Hence
\begin{align*}
    \Delta W
    &=g(T_k-v+u)-g(T_k)\\
    &=g(u\mid T_k-v)-g(v\mid T_k-v)\\
    &\ge g(u\mid T_k)-g(v\mid T_k-v),
\end{align*}
where the inequality follows from $T_k-v\subseteq T_k$ and submodularity.
Since $\lambda\ge0$, equation~\eqref{dir:exchange-f-change} again gives
\[
    \Delta H_\lambda
    \ge\Gamma_{\lambda,\calT}(u,v,k).
\]

For a reassignment of $v\in T_k$ to $T_j$, the underlying base does not
change.  Hence the $f$-term is unchanged, and
\begin{align*}
    \Delta H_\lambda
    &=\lambda\bigl(
        g(T_j+v)-g(T_j)+g(T_k-v)-g(T_k)
      \bigr)\\
    &=\lambda\bigl(
        g(v\mid T_j)-g(v\mid T_k-v)
      \bigr)\\
    &=\Lambda_{\lambda,\calT}(v,j).
\end{align*}
Thus every move selected by the algorithm increases $H_\lambda$ by more than
$\eta$.

It remains to bound the number of iterations.  At initialization, $\calT$ is
a partition of the all-dummy base.  Since both functions ignore dummy
elements,
\[
    H_\lambda(\calT)=f(\varnothing)\ge0.
\]
For any partitioned base, its underlying base becomes an independent set after
the dummy elements are removed, and hence
\[
    f(q(\calT))\le f(\OPT).
\]
Each block $T_b$ is also independent, so \eqref{dir:g-bounds} gives
\[
    g(T_b)-g(\varnothing)\le f(\OPT).
\]
Consequently,
\[
    W(\calT)\le\ell f(\OPT)
    \qquad\text{and}\qquad
    H_\lambda(\calT)\le(1+\lambda\ell)f(\OPT).
\]
Because every iteration increases $H_\lambda$ by more than $\eta$, the number
of iterations is strictly smaller than
\[
    \frac{(1+\lambda\ell)f(\OPT)}{\eta}
    \le
    \frac{4rn(1+\lambda\ell)}{\varepsilon},
\]
where the inequality uses
$\eta=\varepsilon f(S_{\max})/(4r)$ and
$f(\OPT)\le n f(S_{\max})$.  The direction oracle only calls the local search
with $0\le\lambda\le\lambda_{\max}$, so this is a polynomial bound.  Finally,
when the algorithm terminates, no feasible exchange has gain greater than
$\eta$ and no reassignment has gain greater than $\eta$.  This is precisely
the definition of $\eta$-local optimality.
\end{proof}

The locally optimal partitioned base satisfies the tradeoff needed by the
direction oracle.  To obtain it, we average the terminal local-optimality
inequalities over a base-exchange matching and separate their $f$- and
$g$-contributions.  Exchanges involving dummy elements control the
potentially negative marginals of elements common to the two bases.

\begin{lemma}
\label{lem:weighted-exchange}
Let $\calT=(T_1,\ldots,T_\ell)$ be the output of
Algorithm~\ref{alg:extended-local-search} on input $(f,g,\lambda)$, and put
$T=q(\calT)$.
Then
\begin{equation}
\begin{aligned}
    &f(T\cup\OPT)+f(T\cap\OPT)-2f(T)+
      \lambda\left[
        \left(1-\frac1\ell\right)g(\OPT)-g(\varnothing)
        -\left(1+\frac1\ell\right)W(\calT)
      \right]\\
    &\hspace{15mm}\le
      \left(1+\frac1\ell\right)
      \frac{\varepsilon f(S_{\max})}{4}
      +\frac{f(T)}{\ell}.
\end{aligned}
\label{dir:weighted-exchange-bound}
\end{equation}
\end{lemma}

\begin{proof}
By Lemma~\ref{lem:base-exchange-matching}, there is a bijection
$h:\OPT\to T$ such that $T-h(u)+u$ is a base for every $u\in\OPT$, and
$h(u)=u$ for every $u\in\OPT\cap T$.  For $u\in\OPT\setminus T$, define
\[
    A_u=\frac1\ell\sum_{j=1}^{\ell}\Gamma(u,h(u),j).
\]
For $u\in\OPT\cap T_k$, define
\[
    A_u=\frac1\ell
        \sum_{\substack{j=1\\j\ne k}}^{\ell}\Lambda(u,j).
\]
The omitted destination $j=k$ is the null action, whose gain is zero.  Thus
both definitions average over all $\ell$ destination blocks.  Set
\[
    \Sigma=\sum_{u\in\OPT}A_u=\Sigma_f+\lambda\Sigma_g,
\]
where $\Sigma_f$ is the total contribution of the $f$-terms and $\Sigma_g$
is the coefficient of $\lambda$.  Every non-null term is a feasible tested
move and is at most $\eta$ at termination; the null term is zero.  Since
$|\OPT|=r$,
\begin{equation}
    \Sigma\le r\eta
    =\frac{\varepsilon f(S_{\max})}{4}.
    \label{dir:averaged-moves}
\end{equation}

\emph{Contribution of $f$.}
Reassignment and null actions do not change the underlying base, and the
$f$-part of an exchange gain is independent of its destination block.
Therefore
\[
\begin{aligned}
    \Sigma_f
    &=\sum_{u\in\OPT\setminus T}f(u\mid T)
      -\sum_{u\in\OPT\setminus T}f(h(u)\mid T-h(u))=\sum_{u\in\OPT\setminus T}f(u\mid T)
      -\sum_{v\in T\setminus\OPT}f(v\mid T-v),
\end{aligned}
\]
where the second equality uses that $h$ maps $\OPT\setminus T$ bijectively
onto $T\setminus\OPT$.
Adding the elements of $\OPT\setminus T$ successively shows that the first
sum is at least $f(T\cup\OPT)-f(T)$.  Deleting the elements of
$T\setminus\OPT$ successively shows that the second sum is at most
$f(T)-f(T\cap\OPT)$: when $v$ is deleted, its current context is contained
in $T-v$.  Consequently,
\begin{equation}
    \Sigma_f\ge f(T\cup\OPT)+f(T\cap\OPT)-2f(T).
    \label{dir:f-contribution}
\end{equation}

\emph{Contribution of $g$.}
For $v\in T_k$, put $d_v=g(v\mid T_k-v)$.  Expanding the tested gains and
using the bijectivity of $h$ gives
\begin{equation}
\begin{aligned}
    \Sigma_g={}&
    \frac1\ell\sum_{u\in\OPT\setminus T}\sum_{j=1}^{\ell}
        \bigl[g(u\mid T_j)-d_{h(u)}\bigr]+\frac1\ell\sum_{k=1}^{\ell}
        \sum_{u\in\OPT\cap T_k}
        \sum_{\substack{j=1\\j\ne k}}^{\ell}
        \bigl[g(u\mid T_j)-d_u\bigr]\\
    ={}&\frac1\ell\sum_{j=1}^{\ell}\sum_{u\in\OPT}g(u\mid T_j)
        -\sum_{v\in T}d_v
        +\frac1\ell\sum_{u\in\OPT\cap T}d_u.
\end{aligned}
\label{dir:g-exact}
\end{equation}
In the second equality, the final sum corrects the terms with
$u\in\OPT\cap T_k$ and $j=k$: including this destination in the double sum
contributes $g(u\mid T_k)-d_u=-d_u$, whereas the corresponding null action
has value zero.  Successive deletion within each block gives
\begin{equation}
\begin{aligned}
    \sum_{v\in T}d_v
    &=\sum_{j=1}^{\ell}\sum_{v\in T_j}g(v\mid T_j-v)\le\sum_{j=1}^{\ell}\bigl[g(T_j)-g(\varnothing)\bigr]
     =W(\calT).
\end{aligned}
    \label{dir:g-deletion}
\end{equation}
To see the inequality, fix an ordering $v_1,\ldots,v_m$ of $T_j$ and let
$R_i=T_j\setminus\{v_1,\ldots,v_i\}$.  Since
$R_i\subseteq T_j-v_i$, submodularity gives
\[
    g(v_i\mid R_i)\ge g(v_i\mid T_j-v_i),
\]
and summing over $i$ telescopes to
$g(T_j)-g(\varnothing)$.  Applying this argument to every block proves
\eqref{dir:g-deletion}.

Similarly, order $\OPT=\{u_1,\ldots,u_r\}$ and set
$O_i=\{u_1,\ldots,u_i\}$.  Submodularity gives
$g(u_i\mid T_j)\ge g(u_i\mid T_j\cup O_{i-1})$; summing these inequalities
gives
\begin{equation}
    \sum_{u\in\OPT}g(u\mid T_j)
    \ge g(\OPT\cup T_j)-g(T_j).
    \label{dir:g-insertion}
\end{equation}
Because the blocks are pairwise disjoint, repeated submodularity also gives
\begin{equation}
    \sum_{j=1}^{\ell}g(\OPT\cup T_j)
    \ge(\ell-1)g(\OPT)+g(\OPT\cup T)
    \ge(\ell-1)g(\OPT).
    \label{dir:g-disjoint-unions}
\end{equation}
The first inequality follows from the equivalent form of
submodularity
\[
    g(A)+g(B)\ge g(A\cup B)+g(A\cap B).
\]
Indeed, when $A=\OPT\cup\bigcup_{i<j}T_i$ and
$B=\OPT\cup T_j$, pairwise disjointness of the blocks gives
$A\cap B=\OPT$.  Applying this inequality successively for
$j=2,\ldots,\ell$ yields the first inequality in
\eqref{dir:g-disjoint-unions}; its second inequality uses $g\ge0$.  Applying
\eqref{dir:g-deletion}--\eqref{dir:g-disjoint-unions} to
\eqref{dir:g-exact}, and using
$\sum_jg(T_j)=\ell g(\varnothing)+W(\calT)$, yields
\begin{equation}
\begin{aligned}
    \Sigma_g
    &\ge\frac1\ell\sum_{j=1}^{\ell}
        \bigl[g(\OPT\cup T_j)-g(T_j)\bigr]
        -W(\calT)
        +\frac1\ell\sum_{u\in\OPT\cap T}d_u\\
    &\ge\left(1-\frac1\ell\right)g(\OPT)-g(\varnothing)
        -\left(1+\frac1\ell\right)W(\calT)
        +\frac1\ell\sum_{u\in\OPT\cap T}d_u.
\end{aligned}
    \label{dir:g-lower-bound}
\end{equation}

\emph{Correction for common elements.}
Combining \eqref{dir:f-contribution} and \eqref{dir:g-lower-bound} gives
\begin{equation}
\begin{aligned}
    \Sigma\ge{}&f(T\cup\OPT)+f(T\cap\OPT)-2f(T)\\
    &+\lambda\left[
        \left(1-\frac1\ell\right)g(\OPT)-g(\varnothing)
        -\left(1+\frac1\ell\right)W(\calT)
      \right]\\
    &+\frac{\lambda}{\ell}\sum_{u\in\OPT\cap T}d_u.
\end{aligned}
\label{dir:g-common-correction}
\end{equation}
The last sum cannot be discarded because the deletion marginals $d_u$ need
not be non-negative.

Write $a_v=f(v\mid T-v)$ and $a_v^+=\max\{a_v,0\}$.  For every real element
$v\in T_k$, there is an unused dummy element $e$ such that $T-v+e$ is a base.
Both $f$ and $g$ ignore $e$, so local optimality for this exchange gives
\[
    -a_v-\lambda d_v\le\eta,
    \qquad\text{and hence}\qquad
    \lambda d_v\ge-a_v-\eta\ge-a_v^+-\eta.
\]
For a dummy element, $a_v=d_v=0$.  Since $|\OPT\cap T|\le r$, it follows that
\begin{equation}
    \lambda\sum_{u\in\OPT\cap T}d_u
    \ge-\sum_{v\in T}a_v^+-r\eta
    \ge-f(T)-\frac{\varepsilon f(S_{\max})}{4}.
    \label{dir:negative-correction}
\end{equation}
Indeed, let $P=\{v\in T:a_v>0\}$ and delete the elements of $P$ successively.
At each deletion, the current context is contained in $T-v$, so submodularity
implies that the actual loss is at least $a_v$.  Therefore
\[
    \sum_{v\in T}a_v^+
    =\sum_{v\in P}a_v
    \le f(T)-f(T\setminus P)
    \le f(T),
\]
where the last inequality follows from non-negativity.

Substituting \eqref{dir:negative-correction} into
\eqref{dir:g-common-correction} yields
\[
\begin{aligned}
    \Sigma\ge{}&f(T\cup\OPT)+f(T\cap\OPT)-2f(T)\\
    &+\lambda\left[
        \left(1-\frac1\ell\right)g(\OPT)-g(\varnothing)
        -\left(1+\frac1\ell\right)W(\calT)
      \right]\\
    &-\frac{f(T)}{\ell}
      -\frac{\varepsilon f(S_{\max})}{4\ell}.
\end{aligned}
\]
Comparing this lower bound with \eqref{dir:averaged-moves} and rearranging
proves \eqref{dir:weighted-exchange-bound}.
\end{proof}

\subsubsection{Multiplier search and mixing}

The exchange inequality couples $\varepsilon$-stationarity of the induced
guide distribution, as defined in \eqref{main:guide-stationarity}, with the
aggregate block gain.  A single multiplier need not yield exactly the desired
gain, so the direction oracle combines $\eta$-locally optimal partitioned
bases obtained at two nearby multiplier values.  Given a target
$V\in\mathbb R$, define
\begin{equation}
\begin{aligned}
    b
    &=\frac{(\ell-1)V-\ell g(\varnothing)
             -\ell\varepsilon f(S_{\max})}{\ell+1}, \qquad
    \lambda_{\max}
    =\frac{4n}{\varepsilon},
    \qquad
    \tau=\frac{\varepsilon}{4(\ell+4)n}.
\end{aligned}
\label{dir:search-parameters}
\end{equation}
Here $b$ is the mixing threshold for the aggregate block gain,
$\lambda_{\max}$ is the largest multiplier considered, and $\tau$ is the
terminal width of the multiplier interval.
Let $\calT^{\mathrm d}$ be the fixed partition of the all-dummy base.

\begin{algorithm}[tb]
\DontPrintSemicolon
\KwIn{Submodular functions $f,g$ and a target value $V$.}
\KwOut{A direction distribution $D$ supported on at most two partitioned bases.}

Compute $b$, $\lambda_{\max}$, and $\tau$ from
\eqref{dir:search-parameters}\;
$\lambda_l\gets0$ and
$\calT_l\gets\operatorname{ExtendedLocalSearch}(f,g,\lambda_l)$\;
\If{$W(\calT_l)\ge b$}{
    \Return{$\{(\calT_l,1)\}$}\;
}

$\lambda_h\gets\lambda_{\max}$ and
$\calT_h\gets\operatorname{ExtendedLocalSearch}(f,g,\lambda_h)$\;
\If{$W(\calT_h)<b$}{
    \Return{$\{(\calT^{\mathrm d},1)\}$}\;
}

\While{$\lambda_h-\lambda_l>\tau$}{
    $\lambda_c\gets(\lambda_l+\lambda_h)/2$\;
    $\calT_c\gets\operatorname{ExtendedLocalSearch}(f,g,\lambda_c)$\;
    \If{$W(\calT_c)<b$}{
        $(\lambda_l,\calT_l)\gets(\lambda_c,\calT_c)$\;
    }
    \Else{
        $(\lambda_h,\calT_h)\gets(\lambda_c,\calT_c)$\;
    }
}

$p\gets
\bigl(W(\calT_h)-b\bigr)/
\bigl(W(\calT_h)-W(\calT_l)\bigr)$\;
\Return{$\{(\calT_l,p),(\calT_h,1-p)\}$, discarding zero weights}\;

\caption{$\operatorname{DeterministicDirectionOracle}(f,g,V)$}
\label{alg:direction-oracle}
\end{algorithm}

Let $\calT(\lambda)$ denote the output of
$\operatorname{ExtendedLocalSearch}(f,g,\lambda)$.  The search does not assume
that the aggregate block gain $W(\calT(\lambda))$ is monotone in $\lambda$;
it maintains only the endpoint invariant
\begin{equation}
    W(\calT_l)<b\le W(\calT_h).
    \label{dir:bracketing-invariant}
\end{equation}
Whichever side contains the midpoint output can replace the corresponding
endpoint, preserving the invariant while halving the interval.

The analysis first bounds the error in every local-search output and uses
this bound to exclude the all-dummy fallback.  It then mixes the two bracket
endpoints and averages their weighted exchange inequalities.  The mixture
simultaneously provides the stationarity and growth guarantees stated next.

\begin{theorem}
\label{thm:direction-oracle}
Assume \eqref{dir:g-bounds} and $g(\OPT)\ge V$.  Then
Algorithm~\ref{alg:direction-oracle} does not return the all-dummy fallback.
It returns a direction distribution $D$ supported on at most two partitioned
bases such that, for $\calT=(T_1,\ldots,T_\ell)\sim D$ and $T=q(\calT)$,
\begin{align}
    \E_D\bigl[
        f(T\cup\OPT)+f(T\cap\OPT)-2f(T)
    \bigr]
    &\le2\varepsilon f(\OPT),
    \label{dir:averaged-stationarity}\\
    \E_D W(\calT)
    &\ge V-g(\varnothing)
       -\left(\frac4\ell+\varepsilon\right)f(\OPT).
    \label{dir:averaged-growth}
\end{align}
Moreover, $\E_D W(\calT)\ge b$, with equality whenever the returned
distribution uses both endpoints with positive weight.
\end{theorem}

\begin{proof}
Every local-search output satisfies \eqref{dir:weighted-exchange-bound}.
For any such output with underlying base $T$, its right-hand side is bounded
uniformly by
\begin{equation}
\begin{aligned}
    \left(1+\frac1\ell\right)\frac{\varepsilon f(S_{\max})}{4}
      +\frac{f(T)}{\ell}
    &\le\left(1+\frac1\ell\right)
       \frac{\varepsilon f(S_{\max})}{4}
       +\frac{f(\OPT)}{\ell}\\
    &\le\frac{11}{40}\varepsilon f(\OPT)
       +\varepsilon f(\OPT)
     =\frac{51}{40}\varepsilon f(\OPT).
\end{aligned}
\label{dir:local-search-error}
\end{equation}
Here we used $f(S_{\max})\le f(\OPT)$, $\ell\ge10$, and
$1/\ell\le\varepsilon$.  The hypothesis $g(\OPT)\ge V$ and the
definition of $b$ also give
\begin{equation}
\begin{aligned}
    &\left(1-\frac1\ell\right)g(\OPT)-g(\varnothing)
       -\left(1+\frac1\ell\right)b=\left(1-\frac1\ell\right)(g(\OPT)-V)
       +\varepsilon f(S_{\max})
     \ge\varepsilon f(S_{\max}).
\end{aligned}
\label{dir:threshold-slack}
\end{equation}

If $W(\calT_l)\ge b$ at $\lambda_l=0$, the oracle returns
$D=\{(\calT_l,1)\}$.  Equation~\eqref{dir:local-search-error} then gives
\[
\begin{aligned}
    &\E_D\bigl[f(T\cup\OPT)+f(T\cap\OPT)-2f(T)\bigr]\\
    &\qquad=f(q(\calT_l)\cup\OPT)+f(q(\calT_l)\cap\OPT)
      -2f(q(\calT_l))
      \le\frac{51}{40}\varepsilon f(\OPT)
      <2\varepsilon f(\OPT),
\end{aligned}
\]
where the weighted exchange inequality is used at $\lambda_l=0$.
This proves \eqref{dir:averaged-stationarity} for the single-output case.
Otherwise $W(\calT_l)<b$.  Suppose that the output $\calT_h$ at
$\lambda_{\max}$ also had $W(\calT_h)<b$.  Since the coefficient of
$W$ in \eqref{dir:weighted-exchange-bound} is negative,
\eqref{dir:threshold-slack} implies
\[
\begin{aligned}
    &\left(1-\frac1\ell\right)g(\OPT)-g(\varnothing)
       -\left(1+\frac1\ell\right)W(\calT_h)>\left(1-\frac1\ell\right)g(\OPT)-g(\varnothing)
       -\left(1+\frac1\ell\right)b
     \ge\varepsilon f(S_{\max}).
\end{aligned}
\]
Non-negativity of $f$ gives
$f(q(\calT_h)\cup\OPT)+f(q(\calT_h)\cap\OPT)
-2f(q(\calT_h))\ge-2f(\OPT)$.
Consequently, the left-hand side of
\eqref{dir:weighted-exchange-bound} for $\calT_h$ is greater than
\[
\begin{aligned}
    -2f(\OPT)+\lambda_{\max}\varepsilon f(S_{\max})
    &=-2f(\OPT)+4n f(S_{\max})\\
    &\ge2n f(S_{\max})
     >\frac{51}{40}\varepsilon f(\OPT),
\end{aligned}
\]
where $f(\OPT)\le n f(S_{\max})$ and $0<\varepsilon\le1/10$.
This contradicts \eqref{dir:local-search-error}.  Thus
$W(\calT_h)\ge b$, so the all-dummy fallback is not taken.

At each bisection step, the midpoint output replaces the endpoint on the
same side of $b$.  Hence \eqref{dir:bracketing-invariant} persists, and
after $k$ steps the interval has length $\lambda_{\max}/2^k$; choosing
$k\ge\lceil\log_2(\lambda_{\max}/\tau)\rceil$ ensures
$\lambda_h-\lambda_l\le\tau$.  In particular,
$W(\calT_l)<b\le W(\calT_h)$, so the algorithm's weight
satisfies
\begin{equation}
\begin{aligned}
    p
    &=\frac{W(\calT_h)-b}
            {W(\calT_h)-W(\calT_l)}\in[0,1],\\
    pW(\calT_l)+(1-p)W(\calT_h)&=b.
\end{aligned}
    \label{dir:exact-mixing}
\end{equation}

We next put the two exchange inequalities at the common multiplier
$\lambda_l$.  From \eqref{dir:g-bounds}, each block satisfies
$-g(\varnothing)\le g(T_j)-g(\varnothing)\le f(\OPT)$ and
$g(\varnothing)\le f(\OPT)$.  Since $\OPT$ is feasible, the same
assumption gives $g(\OPT)\le g(\varnothing)+f(\OPT)\le2f(\OPT)$.
Therefore
\[
    |W(\calT)|\le\ell f(\OPT)
\]
and, by the triangle inequality,
\begin{equation}
\begin{aligned}
    &\left|\left(1-\frac1\ell\right)g(\OPT)-g(\varnothing)
       -\left(1+\frac1\ell\right)W(\calT)\right|\\
    &\quad\le\left(1-\frac1\ell\right)2f(\OPT)+f(\OPT)
       +\left(1+\frac1\ell\right)\ell f(\OPT)\\
    &\quad\le(\ell+4)f(\OPT)
      \le(\ell+4)n f(S_{\max}).
\end{aligned}
\label{dir:multiplier-coefficient-bound}
\end{equation}
The low-endpoint weighted exchange inequality has multiplier $\lambda_l$.
For the high endpoint, subtracting the difference
$(\lambda_h-\lambda_l)$ times the bracketed coefficient and applying
\eqref{dir:multiplier-coefficient-bound} gives
\begin{equation}
\begin{aligned}
    &f(q(\calT_h)\cup\OPT)+f(q(\calT_h)\cap\OPT)
      -2f(q(\calT_h))\\
    &\quad+\lambda_l\left[
       \left(1-\frac1\ell\right)g(\OPT)-g(\varnothing)
       -\left(1+\frac1\ell\right)W(\calT_h)\right]\\
    &\le\frac{51}{40}\varepsilon f(\OPT)
       +(\lambda_h-\lambda_l)(\ell+4)n f(S_{\max})\\
    &\le\frac{51}{40}\varepsilon f(\OPT)
       +\tau(\ell+4)n f(S_{\max})\\
    &=\frac{51}{40}\varepsilon f(\OPT)
       +\frac{\varepsilon f(S_{\max})}{4},
\end{aligned}
\label{dir:common-multiplier}
\end{equation}
Multiply the low- and high-endpoint inequalities by $p$ and $1-p$ and add.
By \eqref{dir:exact-mixing}, the weighted average of their bracketed
coefficients replaces $W(\calT_l)$ and $W(\calT_h)$ by $b$.
Thus
\begin{equation}
\begin{aligned}
    &\E_D\bigl[f(T\cup\OPT)+f(T\cap\OPT)-2f(T)\bigr]\\
    &\quad+\lambda_l\left[
      \left(1-\frac1\ell\right)g(\OPT)-g(\varnothing)
      -\left(1+\frac1\ell\right)b\right]\\
    &\le\frac{51}{40}\varepsilon f(\OPT)
       +(1-p)\frac{\varepsilon f(S_{\max})}{4}\\
    &\le\frac{61}{40}\varepsilon f(\OPT)
     <2\varepsilon f(\OPT).
\end{aligned}
\label{dir:mixed-stationarity}
\end{equation}
By \eqref{dir:threshold-slack} and $\lambda_l\ge0$, the bracketed term
multiplied by $\lambda_l$ is non-negative.  Dropping it
proves \eqref{dir:averaged-stationarity} for the two-endpoint return.

Finally, the first return has $W(\calT_l)\ge b$, while the mixed return
has $\E_D W(\calT)=b$ by \eqref{dir:exact-mixing}.  Thus in every
non-fallback case,
\begin{equation}
\begin{aligned}
    \E_D W(\calT)
    &\ge b\\
    &=V-g(\varnothing)
      -\frac{2V-g(\varnothing)+\ell\varepsilon f(S_{\max})}{\ell+1}\\
    &\ge V-g(\varnothing)
      -\left(\frac4\ell+\varepsilon\right)f(\OPT).
\end{aligned}
\label{dir:growth-calculation}
\end{equation}
For the last inequality, $V\le g(\OPT)\le2f(\OPT)$,
$g(\varnothing)\ge0$, and $f(S_{\max})\le f(\OPT)$ imply
\[
    \frac{2V-g(\varnothing)+\ell\varepsilon f(S_{\max})}{\ell+1}
    \le\frac{4f(\OPT)+\ell\varepsilon f(\OPT)}{\ell+1}
    \le\left(\frac4\ell+\varepsilon\right)f(\OPT).
\]
Only the upper bound on $V$ is used, so this argument also covers
$V<0$.  Equation~\eqref{dir:growth-calculation} proves
\eqref{dir:averaged-growth}, as well as the stated bound by $b$.
\end{proof}

The $\varepsilon$-stationarity guarantee concerns the distribution of
underlying bases induced by $D$ and need not hold for the underlying base of
either support partitioned base separately.  Accordingly, a recursive call
retains both underlying bases and their weights.  For the initial guide the
auxiliary function is absent, so we specialize the terminal local-optimality
conditions to $g\equiv0$, match the optimum to the returned base, and sum the
matched exchange inequalities.

\begin{corollary}
\label{cor:initial-guide}
Let $\calT_0$ be returned by
$\operatorname{ExtendedLocalSearch}(f,0,0)$ and let $Z_0=q(\calT_0)$.  Then
\[
    f(Z_0\cup\OPT)+f(Z_0\cap\OPT)-2f(Z_0)
    \le\frac{\varepsilon f(S_{\max})}{4}.
\]
Thus the main procedure may take
$\operatorname{InitialGuideDistribution}(f,\calM)=\{(Z_0,1)\}$.
\end{corollary}

\begin{proof}
When $g\equiv0$ and $\lambda=0$, the tested exchange gain is independent of
the destination block.  Apply Lemma~\ref{lem:base-exchange-matching} to
$\OPT$ and $Z_0$, and sum the terminal exchange inequalities for the
$r$ matched pairs.  The insertion-and-deletion argument used in
\eqref{dir:f-contribution} shows that their left-hand side is at least
$f(Z_0\cup\OPT)+f(Z_0\cap\OPT)-2f(Z_0)$, whereas their right-hand side is
$r\eta=\varepsilon f(S_{\max})/4$.
\end{proof}

\subsection{Deterministic guided greedy}
\label{greedy:section}

The procedure maintains a weighted distribution of EME states, forms an
auxiliary growth function from the whole distribution, and calls the
direction oracle in every round.  It rounds all final states into feasible
candidates while retaining every induced guide distribution for possible
recursive calls.

\paragraph{State representation.}\label{greedy:state-representation}
The guided greedy procedure maintains an explicit distribution over partial
growth states.  At the beginning of round $i$, the algorithm maintains a weighted
family of \emph{states}
\[
    \{(w_\omega,Z_\omega,\mathcal L_i^\omega):\omega\in\Omega_i\},
    \qquad w_\omega>0,\quad\sum_{\omega\in\Omega_i}w_\omega=1.
\]
Here $Z_\omega$ is a fixed guide base and $\mathcal L_i^\omega$ is the
list representation of the state's current EME vector
$y_i^\omega$, as defined in
\eqref{eq:eme-list-to-vector}.  The state index also records which
partitioned-base outcome was selected in each preceding round; different
histories remain distinct even if they produce the same guide base and EME
vector.  Initially, each positive-weight base in the input guide distribution
gives one state with its guide weight and an empty list.

A state does not fix the outcomes of its block activations: the listed blocks
are activated independently within that state.  Thus the guide and the EME
vector are coupled through the state index, and the current expected value is
$\sum_{\omega\in\Omega_i}w_\omega F(y_i^\omega)$.  We refer to this weighted family as the \emph{state
distribution}; the algorithm stores all its members.

\paragraph{Guesses for the growth thresholds.}
The thresholds depend on three quantities unavailable to the algorithm:
\[
    f(\OPT),\qquad
    \E_{Z\sim\calZ}f(\OPT\cap Z),\qquad
    \E_{Z\sim\calZ}f(\OPT\cup Z).
\]
We denote their respective guesses by
$\widehat v$, $\widehat v_\cap$, and $\widehat v_\cup$.
The algorithm enumerates
$\widehat v=f(S_{\max})(1+\varepsilon)^j$ for $j$ from zero through the
first index for which this value exceeds $nf(S_{\max})$.  For each
$\widehat v$, it independently enumerates
\begin{equation}
\begin{aligned}
    \widehat v_\cap&=k_\cap\varepsilon\widehat v,
    &k_\cap&\in\left\{0,\ldots,\left\lceil\frac2\varepsilon\right\rceil\right\},\\
    \widehat v_\cup&=k_\cup\varepsilon\widehat v,
    &k_\cup&\in\left\{0,\ldots,\left\lceil\frac3\varepsilon\right\rceil\right\}.
\end{aligned}
\label{greedy:guess-grids}
\end{equation}
These grids do not require knowledge of $\OPT$ and contain a triple satisfying
\begin{equation}
\begin{aligned}
    &(1-\varepsilon)f(\OPT)\le\widehat v\le f(\OPT),\\
    &0\le\E_{Z\sim\calZ}f(\OPT\cap Z)-\widehat v_\cap
      <\varepsilon\widehat v,\\
    &0\le\E_{Z\sim\calZ}f(\OPT\cup Z)-\widehat v_\cup
      <\varepsilon\widehat v.
\end{aligned}
\label{greedy:correct-guesses}
\end{equation}
Indeed, $f(S_{\max})\le f(\OPT)\le nf(S_{\max})$; feasibility of
$\OPT\cap Z$ bounds its value by $f(\OPT)$, while submodularity and
non-negativity give $f(\OPT\cup Z)\le f(\OPT)+f(Z)\le2f(\OPT)$.
We call a run using a triple satisfying \eqref{greedy:correct-guesses}
\emph{correctly guessed}.  This designation is used only in the analysis:
the algorithm executes every enumerated run.

\subsubsection{State update}
\label{greedy:state-update-section}
Each round disables the guide after the phase switch, aggregates all
statewise marginal opportunities into $g_i$, compares them with a guessed
trajectory threshold $V_i$, and uses the direction oracle to append one
block family to every state.  The threshold is a deliberately slackened
approximation of the exact-value threshold $V_i^\star$ defined later.

Fix a guess triple and a round $i$.  The guide is used only before the phase
switch; write
\begin{equation}
    Z_{\omega,i}=
    \begin{cases}
        Z_\omega,&i<i_s,\\
        \varnothing,&i\ge i_s.
    \end{cases}
    \label{greedy:active-guide}
\end{equation}
The current state distribution defines the auxiliary function
\begin{equation}
    g_i(S)=\sum_{\omega\in\Omega_i}w_\omega
        F\bigl(y_i^\omega\vee e_{S\setminus Z_{\omega,i}}\bigr),
    \qquad S\subseteq\mathcal N.
    \label{greedy:auxiliary}
\end{equation}
Each summand can be evaluated from the list representation by appending
$(S\setminus Z_{\omega,i},1)$ to $\mathcal L_i^\omega$.
In particular,
$g_i(\varnothing)=\sum_{\omega\in\Omega_i}w_\omega F(y_i^\omega)$.
In round $i$, the direction oracle is called with $g=g_i$; accordingly,
$W$ denotes the aggregate block gain in \eqref{dir:block-gain} formed from
$g_i$.
The threshold for the fixed guess triple is
\begin{equation}
V_i=
\begin{cases}
\begin{aligned}
    &\widehat v-\widehat v_\cap
     -\bigl(1-(1-\delta)^i\bigr)\widehat v_\cup
     -4\varepsilon\widehat v,
\end{aligned}&i<i_s,\\[3pt]
\begin{aligned}
    &(1-\delta)^{i-i_s}\widehat v\\
    &\quad-\bigl[(1-\delta)^{i-i_s}
        -(1-\delta)^i(1+\delta(i-i_s))\bigr]\widehat v_\cup
        -6\varepsilon\widehat v,
\end{aligned}&i\ge i_s.
\end{cases}
\label{greedy:threshold}
\end{equation}

Let $D_i=\operatorname{DeterministicDirectionOracle}(f,g_i,V_i)$.
For every old state $\omega\in\Omega_i$ and every partitioned base
$\calT=(T_1,\ldots,T_\ell)\in\supp(D_i)$, create a successor indexed by
$(\omega,\calT)$ with
\begin{equation}
\begin{aligned}
    w_{(\omega,\calT)}&=w_\omega D_i(\calT),\\
    Z_{(\omega,\calT)}&=Z_\omega,\\
    \mathcal L_{i+1}^{(\omega,\calT)}
      &=\mathcal L_i^\omega\mathbin{\|}
        \bigl((T_1\setminus Z_{\omega,i},\delta),\ldots,
              (T_\ell\setminus Z_{\omega,i},\delta)\bigr).
\end{aligned}
\label{greedy:state-update}
\end{equation}
Here $\mathbin{\|}$ denotes list concatenation.  The choice of $\calT$
creates a separate state, whereas activation of its listed blocks remains
within the EME vector.  The same $D_i$, computed from the entire current
state distribution, is used for every old state.  Replacing each outcome
$\calT$ of $D_i$ by its underlying base $q(\calT)$ and merging equal bases
gives the guide distribution retained for a possible recursive call.

\begin{algorithm}[tb]
\DontPrintSemicolon
\KwIn{A non-negative submodular function $f$, a matroid $\calM$, and a guide
distribution $\calZ$.}
\KwOut{A family $\mathcal C$ of feasible candidates and a family $\mathcal G$
of guide distributions.}

$\mathcal C\gets\varnothing$ and $\mathcal G\gets\varnothing$\;
\ForEach{guess triple $(\widehat v,\widehat v_\cap,\widehat v_\cup)$
specified above}{
    $\Omega_0\gets\varnothing$\;
    \ForEach{$(Z,\theta)\in\calZ$ with $\theta>0$}{
        Add a state $\omega$ to $\Omega_0$ with
        $w_\omega=\theta$, $Z_\omega=Z$, and
        $\mathcal L_0^\omega=()$\;
    }
    \For{$i=0,\ldots,1/\delta-1$}{
        Define $Z_{\omega,i}$, $g_i$, and $V_i$ by
        \eqref{greedy:active-guide}--\eqref{greedy:threshold}\;
        $D_i\gets\operatorname{DeterministicDirectionOracle}(f,g_i,V_i)$\;
        Add to $\mathcal G$ the guide distribution induced by $D_i$
        through $q$, merging equal bases\;
        $\Omega_{i+1}\gets\varnothing$\;
        \ForEach{$\omega\in\Omega_i$ and $\calT\in\supp(D_i)$}{
            Add the successor $(\omega,\calT)$ to $\Omega_{i+1}$
            using \eqref{greedy:state-update}\;
        }
    }
    \ForEach{$\omega\in\Omega_{1/\delta}$}{
        Apply Theorem~\ref{thm:eme-rounding} to $y_{1/\delta}^\omega$
        and add the returned independent set to $\mathcal C$\;
    }
}
\Return{$(\mathcal C,\mathcal G)$}\;
\caption{$\operatorname{DeterministicGuidedGreedy}(f,\calM,\calZ)$}
\label{alg:guided-greedy}
\end{algorithm}

The recursive role of guided greedy is captured by the following two-outcome
guarantee.  A run either produces a feasible set with the required growth
bound or returns a stationary guide whose expected union value with the
optimum has decreased by a fixed amount.

\begin{theorem}
\label{thm:guided-greedy}
Let $\calZ$ be a guide distribution supported on at most two bases, and let
$(\mathcal C,\mathcal G)$ be the output of
Algorithm~\ref{alg:guided-greedy}.  At least one of the following holds:
\begin{enumerate}[label=(\roman*)]
\item There is a feasible candidate $S\in\mathcal C$ such that
\begin{equation}
\begin{aligned}
f(S)\ge{}&
  \bigl[(2-t_s)e^{t_s-1}-e^{-1}\bigr]f(\OPT)\\
&-\bigl[e^{t_s-1}-e^{-1}\bigr]
  \E_{Z\sim\calZ}f(\OPT\cap Z)\\
&-\left[(2-t_s)e^{t_s-1}
  -\left(2+\frac{(1-t_s)^2}{2}\right)e^{-1}\right]
  \E_{Z\sim\calZ}f(\OPT\cup Z)\\
&-15\varepsilon f(\OPT).
\end{aligned}
\label{greedy:successful-bound}
\end{equation}
\item There is a guide distribution $\calZ'\in\mathcal G$ that is
$\varepsilon$-stationary and satisfies
\begin{equation}
    \E_{Z'\sim\calZ'}f(\OPT\cup Z')
    <\E_{Z\sim\calZ}f(\OPT\cup Z)-\varepsilon f(\OPT).
    \label{greedy:improved-guide}
\end{equation}
\end{enumerate}
\end{theorem}

We first establish feasibility and the gain obtained in one round.  The
trajectory estimates needed to distinguish the two outcomes of
Theorem~\ref{thm:guided-greedy} follow afterward.  For feasibility, we bound
each state's marginal vector by a convex combination of base indicators and
apply EME rounding; the same representation also yields the structural and
numerical bounds on the auxiliary function.

\begin{lemma}
\label{lem:greedy-feasibility}
For every state $\omega\in\Omega_i$,
$\operatorname{Mar}(y_i^\omega)\in P(\calM)$ and
$F(y_i^\omega)\le f(\OPT)$.  Moreover, $g_i$ is non-negative and
submodular, satisfies
\[
    g_i(S)=g_i(S\cap\mathcal N_0)
    \qquad\text{for every }S\subseteq\mathcal N,
\]
and obeys \eqref{dir:g-bounds}.
\end{lemma}

\begin{proof}
Fix a state after $i$ rounds, and let $\calT^{\omega,j}$ be its selected
partitioned base in round $j<i$.  Within a round the blocks are disjoint,
and the sets appended to the list are subsets of those blocks.  Hence an
element is included by that round with probability at most $\delta$.
The union bound gives, coordinatewise,
\begin{equation}
    \operatorname{Mar}(y_i^\omega)
    \le\delta\sum_{j<i}\mathbf 1_{q(\calT^{\omega,j})}
    \in P(\calM).
    \label{greedy:state-feasibility}
\end{equation}
The vector on the right is a convex combination of base indicators and zero,
since $i\delta\le1$; downward closure of $P(\calM)$ proves the first claim.
Theorem~\ref{thm:eme-rounding} then yields
$F(y_i^\omega)\le f(\OPT)$.

By the probabilistic definition of EME, each summand of
\eqref{greedy:auxiliary} is an expectation of the set function
$S\mapsto f(A\cup(S\setminus Z_{\omega,i}))$ over fixed realized sets $A$.
This function is non-negative and submodular, and its value at $S$ equals its
value at $S\cap\mathcal N_0$.  The weighted sum $g_i$ therefore has the same
properties.  Also
\[
    g_i(\varnothing)=\sum_{\omega\in\Omega_i}w_\omega F(y_i^\omega)
    \le f(\OPT).
\]
If $S$ is feasible, then $S\setminus Z_{\omega,i}$ is feasible and,
by submodularity and non-negativity,
\[
    f\bigl(A\cup(S\setminus Z_{\omega,i})\bigr)
    \le f(A)+f(S\setminus Z_{\omega,i})
    \le f(A)+f(\OPT).
\]
Averaging this inequality proves the remaining bound in
\eqref{dir:g-bounds}.
\end{proof}

We next relate a state update to the aggregate block gain returned by the
direction oracle.  Expanding the new EME value over block-activation patterns
and retaining the events with at most one active block gives the basic
increment; averaging over states and directions then permits the use of the
direction-oracle guarantee.

\begin{lemma}
\label{lem:greedy-one-round}
Every round satisfies
\begin{equation}
\begin{aligned}
    &\sum_{\omega\in\Omega_{i+1}}w_\omega F(y_{i+1}^\omega)\ge
      \sum_{\omega\in\Omega_i}w_\omega F(y_i^\omega)
      +\delta\E_{\calT\sim D_i}W(\calT)
      -2\ell^2\delta^2f(\OPT).
\end{aligned}
\label{greedy:one-round}
\end{equation}
If $g_i(\OPT)\ge V_i$, then the guide distribution induced by $D_i$ is
$\varepsilon$-stationary and
\begin{equation}
\begin{aligned}
    &\sum_{\omega\in\Omega_{i+1}}w_\omega F(y_{i+1}^\omega)\ge
      (1-\delta)\sum_{\omega\in\Omega_i}w_\omega F(y_i^\omega)
      +\delta V_i-6\varepsilon\delta f(\OPT).
\end{aligned}
\label{greedy:successful-round}
\end{equation}
\end{lemma}

\begin{proof}
Write the list of an old state $\omega$ as
$\mathcal L_i^\omega=((A_{\omega,1},p_{\omega,1}),\ldots,
(A_{\omega,m_\omega},p_{\omega,m_\omega}))$.
For a fixed $\calT=(T_1,\ldots,T_\ell)$, the list in
\eqref{greedy:state-update} gives the exact expansion
\[
\begin{aligned}
F(y_{i+1}^{(\omega,\calT)})
={}&\sum_{J\subseteq[m_\omega]}
   \left(\prod_{a\in J}p_{\omega,a}
         \prod_{a\in[m_\omega]\setminus J}(1-p_{\omega,a})\right)\\
&\quad\cdot\sum_{H\subseteq[\ell]}
   \delta^{|H|}(1-\delta)^{\ell-|H|}
   f\left(\bigcup_{a\in J}A_{\omega,a}
      \cup\bigcup_{b\in H}(T_b\setminus Z_{\omega,i})\right).
\end{aligned}
\]
Since $f$ is non-negative, the terms with $|H|\ge2$ can be dropped.
The terms with $H=\varnothing$ sum to
$(1-\delta)^\ell F(y_i^\omega)$.  For each $b\in[\ell]$, the term with
$H=\{b\}$ is
\[
\begin{aligned}
&\delta(1-\delta)^{\ell-1}
  \sum_{J\subseteq[m_\omega]}
  \left(\prod_{a\in J}p_{\omega,a}
        \prod_{a\in[m_\omega]\setminus J}(1-p_{\omega,a})\right)
  f\left(\bigcup_{a\in J}A_{\omega,a}
      \cup(T_b\setminus Z_{\omega,i})\right)\\
&\qquad=\delta(1-\delta)^{\ell-1}
  F\bigl(y_i^\omega\vee e_{T_b\setminus Z_{\omega,i}}\bigr).
\end{aligned}
\]
Multiplying by the successor weight $w_\omega D_i(\calT)$, summing over
$\omega$ and $\calT$, and using \eqref{greedy:auxiliary} yields
\begin{equation}
\begin{aligned}
    &\sum_{\omega\in\Omega_{i+1}}w_\omega F(y_{i+1}^\omega)\ge (1-\delta)^\ell
      \sum_{\omega\in\Omega_i}w_\omega F(y_i^\omega)
      +\delta(1-\delta)^{\ell-1}
       \E_{\calT\sim D_i}\sum_{b=1}^{\ell}g_i(T_b).
\end{aligned}
\label{greedy:zero-singleton}
\end{equation}
By \eqref{dir:block-gain}, subtracting the first two terms on the
right-hand side of \eqref{greedy:one-round} from
\eqref{greedy:zero-singleton} leaves at least
\[
\begin{aligned}
    &\bigl[(1-\delta)^\ell-1+\ell\delta\bigr]g_i(\varnothing)-\delta\bigl[1-(1-\delta)^{\ell-1}\bigr]
      \E_{\calT\sim D_i}\sum_{b=1}^{\ell}g_i(T_b).
\end{aligned}
\]
The first term is non-negative.  By Lemma~\ref{lem:greedy-feasibility},
$g_i(T_b)\le2f(\OPT)$, while
$1-(1-\delta)^{\ell-1}\le(\ell-1)\delta$.  The remaining loss is
therefore at most $2\ell(\ell-1)\delta^2 f(\OPT)$, proving
\eqref{greedy:one-round}.

If $g_i(\OPT)\ge V_i$, Theorem~\ref{thm:direction-oracle} gives the
stationarity assertion and
\[
    \E_{\calT\sim D_i}W(\calT)
    \ge V_i-g_i(\varnothing)
       -\left(\frac4\ell+\varepsilon\right)f(\OPT).
\]
Substitution into \eqref{greedy:one-round}, followed by
$1/\ell\le\varepsilon$ and $2\ell^2\delta\le\varepsilon$,
gives \eqref{greedy:successful-round}.
\end{proof}

The one-round estimate is useful when the threshold condition holds.  To
decide whether it holds, we require lower bounds on $g_i(\OPT)$ that do not
assume success in any earlier round.  Before the switch these bounds follow
from Lov\'asz threshold sets disjoint from the guide; after the switch they
follow by expanding the independent block activations and averaging over the
state histories.

\begin{lemma}
\label{lem:greedy-trajectory}
The following inequalities hold for every run, without assuming
$g_j(\OPT)\ge V_j$ in any round.  For $i<i_s$,
\begin{equation}
\begin{aligned}
    g_i(\OPT)\ge{}&f(\OPT)-\E_{Z\sim\calZ}f(\OPT\cap Z)-\bigl(1-(1-\delta)^i\bigr)
      \E_{Z\sim\calZ}f(\OPT\cup Z).
\end{aligned}
\label{greedy:pre-switch-trajectory}
\end{equation}
For $i\ge i_s$,
\begin{equation}
\begin{aligned}
    g_i(\OPT)\ge{}&(1-\delta)^{i-i_s}f(\OPT)\\
    &-\bigl((1-\delta)^{i-i_s}-(1-\delta)^i\bigr)
      \E_{Z\sim\calZ}f(\OPT\cup Z)\\
    &+\delta(1-\delta)^{i-1}
      \sum_{j=i_s}^{i-1}
      \E_{\calT\sim D_j}f(\OPT\cup q(\calT)).
\end{aligned}
\label{greedy:post-switch-trajectory}
\end{equation}
\end{lemma}

\begin{proof}
\emph{Before the switch, $i<i_s$.}
Fix $\omega\in\Omega_i$.  Every set in $\mathcal L_i^\omega$ is disjoint
from $Z_\omega$.  Each element belongs to at most one block per round,
and the listed blocks are activated independently with probability $\delta$.
Consequently,
\[
    \Mar_u(y_i^\omega)=0\quad(u\in Z_\omega),
    \qquad
    \Mar_u(y_i^\omega)\le1-(1-\delta)^i\quad(u\in\mathcal N).
\]
Let $B_t^\omega=\{u\in\mathcal N:\Mar_u(y_i^\omega)>t\}$.
By Lemma~\ref{lem:lovasz-marginal-lower-bound},
\[
\begin{aligned}
F\bigl(y_i^\omega\vee e_{\OPT\setminus Z_\omega}\bigr)
&\ge\int_0^1
  f\bigl((\OPT\setminus Z_\omega)\cup B_t^\omega\bigr)\,dt\\
&=\int_0^{1-(1-\delta)^i}
  f\bigl((\OPT\setminus Z_\omega)\cup B_t^\omega\bigr)\,dt
  +(1-\delta)^i f(\OPT\setminus Z_\omega),
\end{aligned}
\]
where the equality uses $B_t^\omega=\varnothing$ for
$t\ge1-(1-\delta)^i$.
Since $B_t^\omega\cap Z_\omega=\varnothing$, submodularity gives
\[
\begin{aligned}
f\bigl((\OPT\setminus Z_\omega)\cup B_t^\omega\bigr)
\ge{}&f(\OPT\setminus Z_\omega)
  +f(\OPT\cup Z_\omega\cup B_t^\omega)
  -f(\OPT\cup Z_\omega)\\
\ge{}&f(\OPT\setminus Z_\omega)-f(\OPT\cup Z_\omega),
\end{aligned}
\]
where the second inequality follows from non-negativity.  Substituting this
bound into the integral yields
\[
\begin{aligned}
F\bigl(y_i^\omega\vee e_{\OPT\setminus Z_\omega}\bigr)
\ge{}&f(\OPT\setminus Z_\omega)
  -\bigl[1-(1-\delta)^i\bigr]f(\OPT\cup Z_\omega)\\
\ge{}&f(\OPT)-f(\OPT\cap Z_\omega)
  -\bigl[1-(1-\delta)^i\bigr]f(\OPT\cup Z_\omega).
\end{aligned}
\]
The last inequality uses submodularity and $f(\varnothing)\ge0$.
Taking the weighted sum over $\omega\in\Omega_i$, whose guide bases have
distribution $\calZ$, proves \eqref{greedy:pre-switch-trajectory}.

\emph{After the switch, $i\ge i_s$.}
Fix a state $\omega\in\Omega_i$.  For each preceding round $j<i$, write
its selected partitioned base as
\[
    \calT^{\omega,j}=(T_1^{\omega,j},\ldots,T_\ell^{\omega,j}),
    \qquad B_b^{\omega,j}=T_b^{\omega,j}\setminus Z_{\omega,j}.
\]
By the list representation of $y_i^\omega$, all listed blocks are activated
independently with probability $\delta$.  Hence
\begin{equation}
\begin{aligned}
F(y_i^\omega\vee e_{\OPT})
={}&\sum_{H_0,\ldots,H_{i-1}\subseteq[\ell]}
  \left[\prod_{j=0}^{i-1}
    \delta^{|H_j|}(1-\delta)^{\ell-|H_j|}\right]\cdot f\left(\OPT\cup
  \bigcup_{j=0}^{i-1}\bigcup_{b\in H_j}B_b^{\omega,j}\right).
\end{aligned}
\label{greedy:block-activation-expansion}
\end{equation}
For any fixed set $A$, the blocks $B_1^{\omega,j},\ldots,B_\ell^{\omega,j}$
are pairwise disjoint.  The union of $A$ and the independently activated
blocks has marginal probabilities $1$ on $A$, $\delta$ on
$(\bigcup_b B_b^{\omega,j})\setminus A$, and zero elsewhere.
Lemma~\ref{lem:lovasz-marginal-lower-bound} therefore gives
\begin{equation}
\begin{aligned}
\sum_{H\subseteq[\ell]}
  \delta^{|H|}(1-\delta)^{\ell-|H|}
  f\left(A\cup\bigcup_{b\in H}B_b^{\omega,j}\right)&\ge\int_0^\delta
  f\left(A\cup\bigcup_{b=1}^{\ell}B_b^{\omega,j}\right)\,dt
  +\int_\delta^1 f(A)\,dt\\
&=\delta f\left(A\cup\bigcup_{b=1}^{\ell}B_b^{\omega,j}\right)
  +(1-\delta)f(A).
\end{aligned}
\label{greedy:all-or-none}
\end{equation}
Applying \eqref{greedy:all-or-none} successively to the sums over
$H_0,\ldots,H_{i-1}$ in \eqref{greedy:block-activation-expansion}, we obtain
\begin{equation}
\begin{aligned}
F(y_i^\omega\vee e_{\OPT})
\ge{}&\sum_{J\subseteq\{0,\ldots,i-1\}}
  \delta^{|J|}(1-\delta)^{i-|J|}\cdot f\left(\OPT\cup\bigcup_{j\in J}
  \bigl(q(\calT^{\omega,j})\setminus Z_{\omega,j}\bigr)\right).
\end{aligned}
\label{greedy:round-activation-expansion}
\end{equation}

We now bound the terms in \eqref{greedy:round-activation-expansion}.
The term $J=\varnothing$ contributes $(1-\delta)^i f(\OPT)$.
For $\varnothing\ne J\subseteq\{0,\ldots,i_s-1\}$, the union of the
activated sets, denoted by $Y$, is disjoint from $Z_\omega$.
By submodularity and non-negativity,
\[
\begin{aligned}
f(\OPT\cup Y)
&\ge f(\OPT)+f(\OPT\cup Y\cup Z_\omega)-f(\OPT\cup Z_\omega)\\
&\ge f(\OPT)-f(\OPT\cup Z_\omega).
\end{aligned}
\]
The total weight of these terms is
$(1-\delta)^{i-i_s}\bigl[1-(1-\delta)^{i_s}\bigr]$.
For $J=\{j\}$ with $i_s\le j<i$, we have $Z_{\omega,j}=\varnothing$,
so the corresponding term is
\[
    \delta(1-\delta)^{i-1}
    f\bigl(\OPT\cup q(\calT^{\omega,j})\bigr).
\]
Discarding all remaining terms by non-negativity gives
\begin{equation}
\begin{aligned}
F(y_i^\omega\vee e_{\OPT})
\ge{}&(1-\delta)^i f(\OPT)\\
&+(1-\delta)^{i-i_s}\bigl[1-(1-\delta)^{i_s}\bigr]
  \bigl[f(\OPT)-f(\OPT\cup Z_\omega)\bigr]\\
&+\delta(1-\delta)^{i-1}
  \sum_{j=i_s}^{i-1}f\bigl(\OPT\cup q(\calT^{\omega,j})\bigr).
\end{aligned}
\label{greedy:state-trajectory-bound}
\end{equation}

Finally, the state-update rule assigns total weight $D_j(\calT)$ to the
histories selecting $\calT$ in round $j$, and subsequent updates preserve
that total weight.  The guide distribution is likewise preserved.  Thus
\[
    \sum_{\omega\in\Omega_i}w_\omega f(\OPT\cup Z_\omega)
    =\E_{Z\sim\calZ}f(\OPT\cup Z),
\]
and, for every $j<i$,
\[
    \sum_{\omega\in\Omega_i}w_\omega
      f\bigl(\OPT\cup q(\calT^{\omega,j})\bigr)
    =\E_{\calT\sim D_j}f\bigl(\OPT\cup q(\calT)\bigr).
\]
Since $Z_{\omega,i}=\varnothing$, we have
$g_i(\OPT)=\sum_{\omega\in\Omega_i}w_\omega F(y_i^\omega\vee e_{\OPT})$.
Taking the weighted sum of \eqref{greedy:state-trajectory-bound} over
$\omega\in\Omega_i$ and collecting coefficients proves
\eqref{greedy:post-switch-trajectory}.
\end{proof}

We now compare the lower bounds in Lemma~\ref{lem:greedy-trajectory} with the
thresholds $V_i$ used by the algorithm.  The exact-value threshold
$V_i^\star$ is obtained from $V_i$ by replacing the three grid guesses with
the corresponding unknown quantities.  Comparing the formulas before and
after the switch and applying the three guess-error bounds yields the relation
between $V_i^\star$ and $V_i$ stated below.

\begin{lemma}
\label{lem:greedy-thresholds}
Define the exact-value threshold
\begin{equation}
V_i^\star=
\begin{cases}
\begin{aligned}
    &f(\OPT)-\E_{Z\sim\calZ}f(\OPT\cap Z)\\
    &\quad-\bigl(1-(1-\delta)^i\bigr)
      \E_{Z\sim\calZ}f(\OPT\cup Z),
\end{aligned}&i<i_s,\\[3pt]
\begin{aligned}
    &(1-\delta)^{i-i_s}f(\OPT)\\
    &\quad-\bigl[(1-\delta)^{i-i_s}
      -(1-\delta)^i(1+\delta(i-i_s))\bigr]
      \E_{Z\sim\calZ}f(\OPT\cup Z),
\end{aligned}&i\ge i_s.
\end{cases}
\label{greedy:ideal-threshold}
\end{equation}
For a correctly guessed run, $V_i\ge V_i^\star-8\varepsilon f(\OPT)$
for every $i$.  Moreover, $V_i\le V_i^\star$ before the switch and
$V_i\le V_i^\star-2\varepsilon f(\OPT)$ after it.
\end{lemma}

\begin{proof}
Before the switch, subtract the first branch of
\eqref{greedy:ideal-threshold} from \eqref{greedy:threshold}:
\[
\begin{aligned}
V_i-V_i^\star={}&\widehat v-f(\OPT)
 +\bigl(\E_{Z\sim\calZ}f(\OPT\cap Z)-\widehat v_\cap\bigr)\\
&+\bigl(1-(1-\delta)^i\bigr)
  \bigl(\E_{Z\sim\calZ}f(\OPT\cup Z)-\widehat v_\cup\bigr)
 -4\varepsilon\widehat v.
\end{aligned}
\]
Equation~\eqref{greedy:correct-guesses} bounds this difference below by
$-5\varepsilon f(\OPT)$ and above by
$2\varepsilon\widehat v-4\varepsilon\widehat v\le0$.

After the switch, the coefficient of the union-value term in
\eqref{greedy:threshold} is
\[
    (1-\delta)^{i-i_s}
       -(1-\delta)^i(1+\delta(i-i_s)).
\]
Its absolute value is at most one: each of the two terms lies in $[0,1]$,
and the second is at most
$e^{-i\delta}e^{\delta(i-i_s)}\le1$.
Thus \eqref{greedy:correct-guesses} gives
\[
    -8\varepsilon f(\OPT)
    \le V_i-V_i^\star
    \le-5\varepsilon\widehat v
    \le-2\varepsilon f(\OPT),
\]
where the final inequality uses
$\widehat v\ge(1-\varepsilon)f(\OPT)$ and
$\varepsilon\le1/10$.
\end{proof}

The preceding two lemmas turn the first failed threshold into progress in the
recursive search.  A failure before the switch is ruled out directly.  After
the switch, if no earlier direction decreased the union potential, the
post-switch trajectory and the threshold slack contradict the failure.

\begin{lemma}
\label{lem:greedy-failure}
In a correctly guessed run, suppose that $i$ is the first round with
$g_i(\OPT)<V_i$.  Then $i\ge i_s$, and some $i_s\le j<i$ has
\begin{equation}
    \E_{\calT\sim D_j}f(\OPT\cup q(\calT))
    <\E_{Z\sim\calZ}f(\OPT\cup Z)-\varepsilon f(\OPT).
    \label{greedy:failure-decrease}
\end{equation}
The guide distribution induced by $D_j$ is $\varepsilon$-stationary.
\end{lemma}

\begin{proof}
The pre-switch bounds
\eqref{greedy:pre-switch-trajectory} and
Lemma~\ref{lem:greedy-thresholds} imply $g_i(\OPT)\ge V_i$ whenever
$i<i_s$, so the first failure occurs after the switch.
Suppose, toward a contradiction, that every $i_s\le j<i$ satisfies the
reverse weak inequality to \eqref{greedy:failure-decrease}.  Subtracting
\eqref{greedy:ideal-threshold} from
\eqref{greedy:post-switch-trajectory} then gives
\[
\begin{aligned}
g_i(\OPT)-V_i^\star
&\ge\delta(1-\delta)^{i-1}
  \sum_{j=i_s}^{i-1}
  \left[\E_{\calT\sim D_j}f(\OPT\cup q(\calT))
    -(1-\delta)\E_{Z\sim\calZ}f(\OPT\cup Z)\right]\\
&\ge-\delta(1-\delta)^{i-1}(i-i_s)
       \varepsilon f(\OPT)
 \ge-\varepsilon f(\OPT).
\end{aligned}
\]
Here the second inequality uses non-negativity of the union value and the
last uses $\delta(i-i_s)\le1$.  But
$V_i\le V_i^\star-2\varepsilon f(\OPT)$ by
Lemma~\ref{lem:greedy-thresholds}, contradicting
$g_i(\OPT)<V_i$.  Hence \eqref{greedy:failure-decrease} holds for some
$j<i$.  Since $i$ was the first failed round, $g_j(\OPT)\ge V_j$.
Theorem~\ref{thm:direction-oracle} therefore makes the guide distribution
induced by $D_j$ $\varepsilon$-stationary.
\end{proof}

Applying this decrease repeatedly along the recursion tree shows that some
visited run meets every threshold.  Indeed, after each failure we follow the
corresponding child; each such descent lowers the non-negative union
potential by $\varepsilon f(\OPT)$, so this path cannot reach the depth limit
without first encountering a run with no failure.

\begin{corollary}
\label{cor:main-successful-run}
Suppose $r\ge1$ and $f(S_{\max})>0$.
Among the calls visited by Algorithm~\ref{alg:main-deterministic}, there
is one with an $\varepsilon$-stationary guide distribution whose correctly
guessed run satisfies $g_i(\OPT)\ge V_i$ for every $0\le i<1/\delta$.
\end{corollary}

\begin{proof}
The initial guide distribution is $\varepsilon$-stationary by
Corollary~\ref{cor:initial-guide}, since
$\varepsilon f(S_{\max})/4\le2\varepsilon f(\OPT)$.
At each visited call, the guess enumeration includes a correctly guessed
run by \eqref{greedy:correct-guesses}.  If this run has no failed threshold,
the desired conclusion holds.  Otherwise, Lemma~\ref{lem:greedy-failure}
provides an $\varepsilon$-stationary guide distribution
$\calZ'\in\mathcal G$ satisfying
\[
    \E_{Z'\sim\calZ'}f(\OPT\cup Z')
    <\E_{Z\sim\calZ}f(\OPT\cup Z)-\varepsilon f(\OPT).
\]
Whenever the depth limit has not been reached, the main algorithm visits
the call with guide $\calZ'$ because it explores every member of
$\mathcal G$.

Follow such a child in the analysis whenever the selected correctly
guessed run has a failed threshold.  After $d\ge1$ descents, the guide
distribution $\calZ^{(d)}$ on this path would satisfy
\[
\begin{aligned}
    0
    &\le\E_{Z\sim\calZ^{(d)}}f(\OPT\cup Z)\\
    &<\E_{Z\sim\calZ_0}f(\OPT\cup Z)
      -d\varepsilon f(\OPT)\\
    &\le(2-d\varepsilon)f(\OPT).
\end{aligned}
\]
The last inequality follows from submodularity, non-negativity, and
$f(Z)\le f(\OPT)$ for every base $Z$.  Since $f(\OPT)>0$ and
$d_{\max}=\lceil2/\varepsilon\rceil+1$, these inequalities are impossible
for $d=d_{\max}$.  Hence, before the depth limit, a visited call has an
$\varepsilon$-stationary guide and a correctly guessed run with no failed
threshold.
\end{proof}

It remains to evaluate the value accumulated by a run with no failed
threshold.  We iterate the one-round recurrence, compare the resulting
weighted discrete convolution with its continuous integral, and finally
replace the discretized switch time by $t_s$.

\begin{lemma}
\label{lem:greedy-completed-run}
In a correctly guessed run, if $g_i(\OPT)\ge V_i$ for every
$0\le i<1/\delta$, then the final states satisfy the right-hand side of
\eqref{greedy:successful-bound} with
$\sum_{\omega\in\Omega_{1/\delta}}w_\omega F(y_{1/\delta}^\omega)$
in place of $f(S)$.
\end{lemma}

\begin{proof}
Lemma~\ref{lem:greedy-one-round} and
$V_i\ge V_i^\star-8\varepsilon f(\OPT)$ yield
\[
\begin{aligned}
    \sum_{\omega\in\Omega_{i+1}}w_\omega F(y_{i+1}^\omega)\ge(1-\delta)
      \sum_{\omega\in\Omega_i}w_\omega F(y_i^\omega)
      +\delta V_i^\star-14\varepsilon\delta f(\OPT).
\end{aligned}
\]
Iterating, using $F(0)=f(\varnothing)\ge0$ and the fact that the
geometric weights sum to at most one, gives
\begin{equation}
\begin{aligned}
    \sum_{\omega\in\Omega_{1/\delta}}w_\omega F(y_{1/\delta}^\omega)\ge\delta\sum_{i=0}^{1/\delta-1}
       (1-\delta)^{1/\delta-1-i}V_i^\star
       -14\varepsilon f(\OPT).
\end{aligned}
\label{greedy:threshold-convolution}
\end{equation}

To compare this finite sum with the coefficients in
\eqref{greedy:successful-bound}, put $t=i_s\delta$ and, for
$s\in[0,1]$, define
\[
v_t(s)=
\begin{cases}
\begin{aligned}
    &f(\OPT)-\E_{Z\sim\calZ}f(\OPT\cap Z)\\
    &\quad-(1-e^{-s})\E_{Z\sim\calZ}f(\OPT\cup Z),
\end{aligned}&s<t,\\[3pt]
\begin{aligned}
    &e^{-(s-t)}f(\OPT)\\
    &\quad-\bigl[e^{-(s-t)}-e^{-s}(1+s-t)\bigr]
      \E_{Z\sim\calZ}f(\OPT\cup Z),
\end{aligned}&s\ge t.
\end{cases}
\]
The function $v_t$ is the continuous counterpart of
\eqref{greedy:ideal-threshold} at the actual switch time $t$.  For every
integer $0\le k\le1/\delta$,
\begin{equation}
    0\le e^{-k\delta}-(1-\delta)^k\le\delta.
    \label{greedy:power-error}
\end{equation}
Indeed,
$0\le-\log(1-\delta)-\delta\le\delta^2/[2(1-\delta)]$;
use $1-e^{-x}\le x$, $k\delta\le1$, and $\delta\le1/2$.
Equation~\eqref{greedy:power-error} and
$\E_{Z\sim\calZ}f(\OPT\cap Z)\le f(\OPT)$,
$\E_{Z\sim\calZ}f(\OPT\cup Z)\le2f(\OPT)$ imply
\begin{equation}
    |V_i^\star-v_t(i\delta)|\le7\delta f(\OPT),
    \qquad |V_i^\star|\le4f(\OPT).
    \label{greedy:threshold-approximation}
\end{equation}
To verify the first bound, write $m=i-i_s$ after the switch.  Before the
switch, only the coefficient of the union expectation changes, so the error
is at most $2\delta f(\OPT)$.  After the switch, the coefficient of
$f(\OPT)$ changes by at most $\delta$, while
\[
\begin{aligned}
\left|(1-\delta)^m-(1-\delta)^i(1+\delta m)
  -e^{-m\delta}+e^{-i\delta}(1+\delta m)\right|\le\delta+(1+\delta m)\delta\le3\delta.
\end{aligned}
\]
Since the union expectation is at most $2f(\OPT)$, the total post-switch
error is at most $7\delta f(\OPT)$.  The bound $|V_i^\star|\le4f(\OPT)$
follows directly before the switch; after the switch, both
$(1-\delta)^m$ and $(1-\delta)^i(1+\delta m)$ lie in $[0,1]$.
Also, \eqref{greedy:power-error} and $e^\delta-1\le2\delta$ give
\[
    \left|(1-\delta)^{1/\delta-1-i}-e^{i\delta-1}\right|
    \le3\delta.
\]
Indeed, with $k=1/\delta-1-i$, the left-hand side is at most
$|e^{-k\delta}-(1-\delta)^k|+(e^\delta-1)e^{i\delta-1}\le3\delta$.
The function $v_t$ has absolute value at most $4f(\OPT)$ before the switch
and at most $3f(\OPT)$ after it.  Its derivative with respect to $s$ has
absolute value at most $2f(\OPT)$ before the switch and at most
$5f(\OPT)$ after it.  Hence $s\mapsto e^{s-1}v_t(s)$ is
$8f(\OPT)$-Lipschitz on each side.  Since $t=i_s\delta$ lies on the
summation grid, summing the standard $L\delta^2/2$ error over the
$1/\delta$ grid intervals bounds the left Riemann-sum error by
$4\delta f(\OPT)$.  The weight error contributes at most
$3\delta\cdot4f(\OPT)=12\delta f(\OPT)$ after summation, and the threshold
error contributes at most $7\delta f(\OPT)$.  Combining these estimates with
\eqref{greedy:threshold-approximation} gives
\begin{equation}
\begin{aligned}
    \left|\delta\sum_{i=0}^{1/\delta-1}
       (1-\delta)^{1/\delta-1-i}V_i^\star
       -\int_0^1e^{s-1}v_t(s)\,ds\right|\le(12+7+4)\delta f(\OPT)
      =23\delta f(\OPT).
\end{aligned}
\label{greedy:integral-error}
\end{equation}

Integrating the two branches directly gives
\begin{equation}
\begin{aligned}
\int_0^1 e^{s-1}v_t(s)\,ds
={}&\bigl[(2-t)e^{t-1}-e^{-1}\bigr]f(\OPT)\\
&-\bigl[e^{t-1}-e^{-1}\bigr]
  \E_{Z\sim\calZ}f(\OPT\cap Z)\\
&-\left[(2-t)e^{t-1}
  -\left(2+\frac{(1-t)^2}{2}\right)e^{-1}\right]
  \E_{Z\sim\calZ}f(\OPT\cup Z).
\end{aligned}
\label{greedy:continuous-integral}
\end{equation}
Finally, $0\le t-t_s<\delta$.  On the interval where the two switch
times assign different branches, the absolute values of the branches sum
to at most $7f(\OPT)$ by the branch bounds above.  To check the second
part, write $U=\E_{Z\sim\calZ}f(\OPT\cup Z)$.  For $s\ge\sigma$, the
post-switch branch satisfies
\[
    \frac{\partial v_\sigma(s)}{\partial\sigma}
    =e^{\sigma-s}f(\OPT)-\bigl(e^{\sigma-s}+e^{-s}\bigr)U,
\]
whose absolute value is at most $f(\OPT)+2U\le5f(\OPT)$.
Thus the interval on which the branch choices differ contributes at most
$7(t-t_s)f(\OPT)$, and their common post-switch interval contributes at
most $5(t-t_s)f(\OPT)$.  Hence replacing $t$ by $t_s$ in
\eqref{greedy:continuous-integral} changes its value by at most
$12\delta f(\OPT)$.  Equations~\eqref{greedy:threshold-convolution}
and~\eqref{greedy:integral-error} therefore yield the desired expression
with error $(14\varepsilon+35\delta)f(\OPT)$.  Since
$\delta\le\varepsilon/(2\ell^2)$ and $\ell\ge10$,
$35\delta<\varepsilon$, proving the stated
$15\varepsilon f(\OPT)$ error.
\end{proof}

The preceding lemmas complete the analysis of
Algorithm~\ref{alg:guided-greedy}.  We select a correctly guessed run and
split according to whether it has a first failed threshold.  A failure gives
the improved guide in outcome~(ii), whereas a completed run gives
outcome~(i) after deterministic EME rounding.

\begin{proof}[Proof of Theorem~\ref{thm:guided-greedy}]
By \eqref{greedy:correct-guesses}, the algorithm executes a correctly
guessed run.  If this run has a first round with $g_i(\OPT)<V_i$,
Lemma~\ref{lem:greedy-failure} supplies an earlier distribution $D_j$
whose induced guide distribution belongs to $\mathcal G$, is
$\varepsilon$-stationary, and satisfies
\eqref{greedy:improved-guide}.

Otherwise $g_i(\OPT)\ge V_i$ in every round, so
Lemma~\ref{lem:greedy-completed-run} bounds the weighted EME value of the
final states by the right-hand side of
\eqref{greedy:successful-bound}.  For each final state $\omega$,
Algorithm~\ref{alg:guided-greedy} retains an independent set $S_\omega$
with $f(S_\omega)\ge F(y_{1/\delta}^\omega)$ by
Theorem~\ref{thm:eme-rounding}.  Since the final state weights sum to one,
\[
    \max_{S\in\mathcal C}f(S)
    \ge\sum_{\omega\in\Omega_{1/\delta}}w_\omega f(S_\omega)
    \ge\sum_{\omega\in\Omega_{1/\delta}}w_\omega
        F(y_{1/\delta}^\omega).
\]
This proves \eqref{greedy:successful-bound}.
\end{proof}

%% file: sections/04-analysis.tex
\section{Analysis and implementation}
\label{sec:analysis}

\subsection{Approximation guarantee}
\label{sec:approximation-analysis}
Corollary~\ref{cor:main-successful-run} identifies a recursive call with an
$\varepsilon$-stationary guide and a correctly guessed run meeting every
threshold.  Lemma~\ref{lem:greedy-completed-run} and deterministic EME
rounding give a candidate satisfying \eqref{greedy:successful-bound} at
that call.  We now apply the box guarantee to the same guide distribution
and combine the two bounds.  For the box candidate, we average the
single-base guarantee over the guide distribution and use stationarity to
replace the expected guide value by the intersection and union potentials.

\begin{lemma}
\label{lem:box-distribution}
Let $\calZ$ be an $\varepsilon$-stationary guide distribution, and let
$\mathcal C_{\mathrm{box}}$ contain the output of
$\operatorname{DeterministicBox}(Z,\rho)$ for every positive-weight base
$Z$ in its support.  Then
\begin{equation}
\begin{aligned}
\max_{S\in\mathcal C_{\mathrm{box}}}f(S)
\ge{}\frac{4\rho+\rho^2}{2(\rho+1)^2}
       \E_{Z\sim\calZ}f(\OPT\cap Z)
+\frac{\rho^2}{2(\rho+1)^2}
       \E_{Z\sim\calZ}f(\OPT\cup Z)
       -\varepsilon f(\OPT).
\end{aligned}
\label{analysis:box-distribution-bound}
\end{equation}
\end{lemma}

\begin{proof}
Apply Lemma~\ref{box:guarantee} to each support base with
$Q=\OPT\cap Z$.  The largest returned value is at least the weighted
average of those values, so
\[
\begin{aligned}
\max_{S\in\mathcal C_{\mathrm{box}}}f(S)
\ge{}\frac{2\rho}{(\rho+1)^2}
       \E_{Z\sim\calZ}f(\OPT\cap Z)
+\frac{f(\varnothing)}{(\rho+1)^2}
 +\frac{\rho^2}{(\rho+1)^2}\E_{Z\sim\calZ}f(Z).
\end{aligned}
\]
By \eqref{main:guide-stationarity},
\[
    \E_{Z\sim\calZ}f(Z)
    \ge\frac12\E_{Z\sim\calZ}
          \bigl[f(\OPT\cap Z)+f(\OPT\cup Z)\bigr]
       -\varepsilon f(\OPT).
\]
Substitution gives \eqref{analysis:box-distribution-bound}, using
$f(\varnothing)\ge0$ and $\rho^2/(\rho+1)^2\le1$.
\end{proof}

We can now prove Theorem~\ref{thm:main-approximation} by combining the two
candidates.  We use
Corollary~\ref{cor:main-successful-run} to locate one recursive call at which
both bounds apply, take the same fixed convex combination as in the continuous
analysis, and verify its three coefficients.  The rank-zero and zero-value
instances are handled separately at the end.
\begin{proof}[Proof of the approximation guarantee in
Theorem~\ref{thm:main-approximation}]
First consider $r\ge1$ and $f(S_{\max})>0$.  By
Corollary~\ref{cor:main-successful-run}, the main algorithm visits a call
with an $\varepsilon$-stationary guide distribution $\calZ$ and a correctly
guessed run satisfying $g_i(\OPT)\ge V_i$ in every round.
Lemma~\ref{lem:greedy-completed-run} bounds the weighted EME value of its
final states by the right-hand side of \eqref{greedy:successful-bound}.
Theorem~\ref{thm:eme-rounding} produces a saved candidate of at least this
value, since the maximum rounded value is at least the weighted average.
Lemma~\ref{lem:box-distribution} applies to the box candidates saved at
the same call.  Each ancestor returns a maximum-value set among its own
candidates and recursive outputs, so the main output has at least the
value of either candidate.  Its value is therefore at least any convex
combination of their two lower bounds.

Take the box-bound weight to be $\alpha=0.1974$ and the
guided greedy bound weight to be $1-\alpha$.  Collecting the terms and
replacing the error
$\alpha\varepsilon+(1-\alpha)15\varepsilon$ by the weaker
$15\varepsilon$ gives
\begin{equation}
\begin{aligned}
f(S)\ge{}&
  \left\{(1-\alpha)
    \bigl[(2-t_s)e^{t_s-1}-e^{-1}\bigr]
    -15\varepsilon\right\}f(\OPT)\\
&+\left\{\alpha\frac{4\rho+\rho^2}{2(\rho+1)^2}
   -(1-\alpha)\bigl[e^{t_s-1}-e^{-1}\bigr]\right\}
   \E_{Z\sim\calZ}f(\OPT\cap Z)\\
&+\left\{\alpha\frac{\rho^2}{2(\rho+1)^2}
   -(1-\alpha)\left[(2-t_s)e^{t_s-1}
       -\left(2+\frac{(1-t_s)^2}{2}\right)e^{-1}\right]
  \right\}\E_{Z\sim\calZ}f(\OPT\cup Z).
\end{aligned}
\label{analysis:combined-bound}
\end{equation}
The three coefficient expressions in \eqref{analysis:combined-bound},
before the $15\varepsilon$ error, coincide with those in
\cite[Section~3.1]{BuchbinderF2024DR}, with our $\rho$ corresponding to its parameter
$r$.  For $\alpha=0.1974$, $\rho=2.22$, and $t_s=0.3682$, the coefficient
calculation there shows that the intersection and union coefficients are
non-negative, while the coefficient of $f(\OPT)$ before the error is at
least $0.40101$.  Since both expectations are non-negative, we obtain
\[
    f(S)\ge(0.40101-15\varepsilon)f(\OPT)
    \ge0.401f(\OPT),
\]
where the last inequality uses $\varepsilon=10^{-8}$.

If $r=0$, the empty set is the only feasible set and is returned.  If
$f(S_{\max})=0$, the bound $f(\OPT)\le nf(S_{\max})$ from the
preliminaries implies $f(\OPT)=0$; returning the empty set meets the
claimed non-strict ratio.  Removing dummy elements from a returned set
preserves its value and feasibility for the original matroid.
\end{proof}

%% file: sections/04-complexity.tex
\subsection{Oracle complexity and implementation}
\label{sec:complexity}
We count value- and independence-oracle queries, treating $\varepsilon$ as
fixed.  The rank and $S_{\max}$ in Algorithm~\ref{alg:main-deterministic}
are found using polynomially many independence and value queries.  If the
rank is zero, the empty set is the only feasible set.  If
$f(S_{\max})=0$, the bound $f(S)\le nf(S_{\max})$ for feasible $S$
shows that the empty set is optimal.  Hence all remaining calls satisfy
$r\ge1$ and $f(S_{\max})>0$.
An independence test in the dummy-augmented matroid reduces to checking
the set's cardinality and testing its non-dummy part in the original
matroid.

For $\lambda\le\lambda_{\max}=4n/\varepsilon$, an extended local search
starts at the all-dummy partitioned base with
$H_\lambda=f(\varnothing)\ge0$.  At any partitioned base,
\eqref{dir:g-bounds} gives
\[
    H_\lambda(\calT)
    =f(q(\calT))+\lambda W(\calT)
    \le(1+\lambda\ell)f(\OPT).
\]
Each accepted move raises this objective by more than
$\eta=\varepsilon f(S_{\max})/(4r)$.  Since
$f(\OPT)\le nf(S_{\max})$, the number of accepted moves is
\[
    O\!\left(\frac{nr(1+\lambda\ell)}{\varepsilon}\right)
    =O_\varepsilon(n^2r).
\]
One scan examines $O(\ell nr)$ exchanges and reassignments, each using
a constant number of evaluations of $f$ and $g$ and independence tests.
Thus a local search uses $O_\varepsilon(n^3r^2)=O_\varepsilon(n^5)$
such evaluations and queries.  The multiplier interval has length
$4n/\varepsilon$ and is refined to width at most
$\varepsilon/[4(\ell+4)n]$, so
Algorithm~\ref{alg:direction-oracle} performs
$O_\varepsilon(\log(2n))$ local searches.

At round $i$ of a guided greedy run, there are at most $2^{i+1}$ states,
since the input guide distribution and every direction distribution have
support at most two.  Each state's list has at most $i\ell$ entries.
Enumerating their activation patterns evaluates each summand of
\eqref{greedy:auxiliary} exactly; one evaluation of $g_i$ therefore uses
at most $2^{1/\delta+1+\ell/\delta}$ value queries, a constant for fixed
$\varepsilon$.  The rounding theorem uses
$O(n^2 2^{\fracnum(y)})$ value queries and
$O(n^5\log^2(2n))$ independence queries per final state; here
$\fracnum(y)\le\ell/\delta$.  Consequently, one guided greedy run uses
$O_\varepsilon(n^5\log^2(2n))$ oracle queries.  A box call uses
$O(n^2)$ value queries: its support grows by at most two states per
processed element.

There are $O_\varepsilon(\log(2n))$ enumerated guess triples per
recursive call.  Algorithm~\ref{alg:guided-greedy} stores at most
$1/\delta$ guide distributions per triple, so each call has
$B=O_\varepsilon(\log(2n))$ recursive children.  The guess enumeration
multiplies the cost of one run by another $O_\varepsilon(\log(2n))$, so the
nonrecursive cost of a call is
$C=O_\varepsilon(n^5\log^3(2n))$.  With depth
$d_{\max}=\lceil2/\varepsilon\rceil+1$, there are at most
$1+B+\cdots+B^{d_{\max}}=O_\varepsilon(\log^{d_{\max}}(2n))$ calls.
Multiplying this node count by $C$ gives
\[
    O_\varepsilon\!\left(
        n^5\log^{d_{\max}+3}(2n)
    \right)
\]
oracle queries for Algorithm~\ref{alg:main-deterministic}.  Since
$d_{\max}$ is constant when $\varepsilon$ is fixed and
$\log^{d_{\max}}(2n)=O_\varepsilon(n)$, the simpler conservative bound
\begin{equation}
    O_\varepsilon\bigl(n^6\log^3(2n)\bigr)
    \label{complexity:oracle-bound}
\end{equation}
also holds.  The larger logarithmic exponent, compared with an
implementation that fixes $\widehat v$ throughout a recursion tree,
reflects the explicit re-enumeration of all three guesses at every node
in Algorithm~\ref{alg:guided-greedy}.

All enumeration orders and tie-breaking rules are fixed, and every
probability mass is stored and manipulated explicitly; no sampling is
performed.  
Together with the approximation analysis above, this completes the
proof of Theorem~\ref{thm:main-approximation}.

%% file: sections/05-conclusion.tex
\section{Conclusion}
\label{sec:conclusion}

In this work, we developed a deterministic $0.401$-approximation algorithm for maximizing a nonnegative, non-monotone submodular function subject to a matroid constraint, matching the approximation guarantee of the randomized algorithm of \cite{BuchbinderF2024DR} and improving the previous deterministic $(0.385-\varepsilon)$ approximation algorithm of \cite{ChenGLSZ2026}. Our algorithm provides a deterministic implementation of their recursive guided framework using polynomially many oracle queries.

The main technical point is that exact evaluation through the extended
multilinear extension is not by itself sufficient for this
derandomization. The growth guarantee and the stationarity certificate
used to construct recursive guides must hold simultaneously for the
same distribution. Our direction oracle achieves this through an
explicit distribution supported on at most two partitioned bases:
their blocks specify the EME growth direction, while the distribution
of their underlying bases supplies the recursive guides. Together with
the bounded-support representation of EME states, this allows all
expectations and marginal values in the recursive algorithm to be
evaluated exactly and the terminal states to be rounded
deterministically.

The dependence of the query complexity on the internal accuracy
parameters remains large. Improving this dependence and extending the
finite-support direction-oracle approach to other recursive
approximation frameworks are natural directions for future work.

\paragraph{Statement on AI use.}
The main algorithmic idea underlying this work was proposed by
OpenAI's GPT-6 Astra model. The authors subsequently developed this
idea into the formal algorithms and proofs presented in the paper,
supplied and refined the technical details, and independently verified
the correctness of all mathematical arguments. The model also assisted
with the exposition of the manuscript. The authors reviewed the final
manuscript and take full responsibility for its contents.

\paragraph{Acknowledges.}
This work was supported in part by the National Natural Science Foundation of China Grants Number No. 62325210, 12501450 and the China Postdoctoral Science Foundation under Grant Number No. 2026T190829.

%% file: references.bib
@misc{ChenGLSZ2026,
  author = {Chen, Shengminjie and Gao, Yiwei and Lin, Kaifeng and Sun, Xiaoming and Zhang, Jialin},
  title = {Deterministic Algorithm for Non-monotone Submodular Maximization under Matroid and Knapsack Constraints},
  year = {2026},
  eprint = {2603.11996},
  archivePrefix = {arXiv},
  primaryClass = {cs.DS},
  doi = {10.48550/arXiv.2603.11996},
  url = {https://arxiv.org/abs/2603.11996},
  howpublished = {arXiv:2603.11996}
}

@article{LeeMNS2010,
  author = {Lee, Jon and Mirrokni, Vahab S. and Nagarajan, Viswanath and Sviridenko, Maxim},
  title = {Maximizing Nonmonotone Submodular Functions under Matroid or Knapsack Constraints},
  journal = {SIAM Journal on Discrete Mathematics},
  volume = {23},
  number = {4},
  pages = {2053--2078},
  year = {2010},
  doi = {10.1137/090750020},
  url = {https://doi.org/10.1137/090750020}
}

@article{SunZZZ2024,
  author = {Sun, Xiaoming and Zhang, Jialin and Zhang, Shuo and Zhang, Zhijie},
  title = {Improved Deterministic Algorithms for Non-monotone Submodular Maximization},
  journal = {Theoretical Computer Science},
  volume = {984},
  pages = {114293},
  year = {2024},
  doi = {10.1016/j.tcs.2023.114293},
  eprint = {2208.14388},
  archivePrefix = {arXiv},
  url = {https://doi.org/10.1016/j.tcs.2023.114293}
}

@inproceedings{BuchbinderF2024Deterministic,
  author = {Buchbinder, Niv and Feldman, Moran},
  title = {Deterministic Algorithm and Faster Algorithm for Submodular Maximization Subject to a Matroid Constraint},
  booktitle = {Proceedings of the 65th IEEE Annual Symposium on Foundations of Computer Science (FOCS)},
  pages = {700--712},
  publisher = {IEEE},
  year = {2024},
  doi = {10.1109/FOCS61266.2024.00050},
  eprint = {2408.03583},
  archivePrefix = {arXiv},
  url = {https://doi.org/10.1109/FOCS61266.2024.00050}
}

@inproceedings{BuchbinderF2025Extension,
  author = {Buchbinder, Niv and Feldman, Moran},
  title = {Extending the Extension: Deterministic Algorithm for Non-monotone Submodular Maximization},
  booktitle = {Proceedings of the 57th Annual ACM Symposium on Theory of Computing (STOC)},
  pages = {1130--1141},
  publisher = {Association for Computing Machinery},
  year = {2025},
  doi = {10.1145/3717823.3718120},
  eprint = {2409.14325},
  archivePrefix = {arXiv},
  url = {https://arxiv.org/abs/2409.14325},
  note = {Full version: arXiv:2409.14325}
}

@inproceedings{BuchbinderF2024DR,
  author = {Buchbinder, Niv and Feldman, Moran},
  title = {Constrained Submodular Maximization via New Bounds for {DR}-Submodular Functions},
  booktitle = {Proceedings of the 56th Annual ACM Symposium on Theory of Computing (STOC)},
  pages = {1820--1831},
  publisher = {Association for Computing Machinery},
  year = {2024},
  doi = {10.1145/3618260.3649630},
  eprint = {2311.01129},
  archivePrefix = {arXiv},
  url = {https://arxiv.org/abs/2311.01129},
  note = {Full version: arXiv:2311.01129}
}

@article{BuchbinderFNS2015,
  author = {Buchbinder, Niv and Feldman, Moran and Naor, Joseph and Schwartz, Roy},
  title = {A Tight Linear Time {$(1/2)$}-Approximation for Unconstrained Submodular Maximization},
  journal = {SIAM Journal on Computing},
  volume = {44},
  number = {5},
  pages = {1384--1402},
  year = {2015},
  doi = {10.1137/130929205},
  url = {https://doi.org/10.1137/130929205}
}

@article{BuchbinderF2018,
  author = {Buchbinder, Niv and Feldman, Moran},
  title = {Deterministic Algorithms for Submodular Maximization Problems},
  journal = {ACM Transactions on Algorithms},
  volume = {14},
  number = {3},
  pages = {32:1--32:20},
  year = {2018},
  doi = {10.1145/3184990},
  eprint = {1508.02157},
  archivePrefix = {arXiv},
  url = {https://doi.org/10.1145/3184990}
}

@article{Brualdi1969,
  author = {Brualdi, Richard A.},
  title = {Comments on Bases in Dependence Structures},
  journal = {Bulletin of the Australian Mathematical Society},
  volume = {1},
  number = {2},
  pages = {161--167},
  year = {1969},
  doi = {10.1017/S000497270004140X},
  url = {https://doi.org/10.1017/S000497270004140X}
}

@book{Oxley2011,
  author = {Oxley, James},
  title = {Matroid Theory},
  edition = {2},
  publisher = {Oxford University Press},
  year = {2011},
  isbn = {9780198566946},
  doi = {10.1093/acprof:oso/9780198566946.001.0001},
  url = {https://academic.oup.com/book/34846}
}

@book{Schrijver2003,
  author = {Schrijver, Alexander},
  title = {Combinatorial Optimization: Polyhedra and Efficiency},
  series = {Algorithms and Combinatorics},
  volume = {24},
  publisher = {Springer},
  year = {2003},
  isbn = {9783540443896},
  url = {https://link.springer.com/book/9783540443896}
}

@article{FeigeMV2011,
  author = {Feige, Uriel and Mirrokni, Vahab S. and Vondr{\'a}k, Jan},
  title = {Maximizing Non-monotone Submodular Functions},
  journal = {SIAM Journal on Computing},
  volume = {40},
  number = {4},
  pages = {1133--1153},
  year = {2011},
  doi = {10.1137/090779346},
  url = {https://doi.org/10.1137/090779346}
}

@inproceedings{HanCCW2020,
  author = {Han, Kai and Cao, Zongmai and Cui, Shuang and Wu, Benwei},
  title = {Deterministic Approximation for Submodular Maximization over a Matroid in Nearly Linear Time},
  booktitle = {Advances in Neural Information Processing Systems},
  volume = {33},
  year = {2020},
  eprint = {2010.11420},
  archivePrefix = {arXiv},
  url = {https://proceedings.neurips.cc/paper_files/paper/2020/hash/05128e44e27c36bdba71221bfccf735d-Abstract.html}
}

@inproceedings{ChenNPK2024,
  author = {Chen, Yixin and Nath, Ankur and Peng, Chunli and Kuhnle, Alan},
  title = {Discretely Beyond {$1/e$}: Guided Combinatorial Algorithms for Submodular Maximization},
  booktitle = {Advances in Neural Information Processing Systems},
  volume = {37},
  year = {2024},
  eprint = {2405.05202},
  archivePrefix = {arXiv},
  url = {https://proceedings.neurips.cc/paper_files/paper/2024/hash/c4e40d310af1be55f2c0f32cb16a9a16-Abstract-Conference.html}
}

@article{NemhauserWF1978,
  author = {Nemhauser, George L. and Wolsey, Laurence A. and Fisher, Marshall L.},
  title = {An Analysis of Approximations for Maximizing Submodular Set Functions---{I}},
  journal = {Mathematical Programming},
  volume = {14},
  pages = {265--294},
  year = {1978},
  doi = {10.1007/BF01588971},
  url = {https://doi.org/10.1007/BF01588971}
}

@article{Sviridenko2004,
  author = {Sviridenko, Maxim},
  title = {A Note on Maximizing a Submodular Set Function Subject to a Knapsack Constraint},
  journal = {Operations Research Letters},
  volume = {32},
  number = {1},
  pages = {41--43},
  year = {2004},
  doi = {10.1016/S0167-6377(03)00062-2},
  url = {https://doi.org/10.1016/S0167-6377(03)00062-2}
}

@article{CalinescuCPV2011,
  author = {C{\u a}linescu, Gruia and Chekuri, Chandra and P{\'a}l, Martin and Vondr{\'a}k, Jan},
  title = {Maximizing a Monotone Submodular Function Subject to a Matroid Constraint},
  journal = {SIAM Journal on Computing},
  volume = {40},
  number = {6},
  pages = {1740--1766},
  year = {2011},
  doi = {10.1137/080733991},
  url = {https://doi.org/10.1137/080733991}
}

@inproceedings{GuptaRST2010,
  author = {Gupta, Anupam and Roth, Aaron and Schoenebeck, Grant and Talwar, Kunal},
  title = {Constrained Non-monotone Submodular Maximization: Offline and Secretary Algorithms},
  booktitle = {Internet and Network Economics (WINE)},
  series = {Lecture Notes in Computer Science},
  volume = {6484},
  pages = {246--257},
  publisher = {Springer},
  year = {2010},
  doi = {10.1007/978-3-642-17572-5_20},
  eprint = {1003.1517},
  archivePrefix = {arXiv},
  url = {https://arxiv.org/abs/1003.1517}
}

@inproceedings{PhamTHT2023,
  author = {Pham, Canh V. and Tran, Tan D. and Ha, Dung T. K. and Thai, My T.},
  title = {Linear Query Approximation Algorithms for Non-monotone Submodular Maximization under Knapsack Constraint},
  booktitle = {Proceedings of the Thirty-Second International Joint Conference on Artificial Intelligence (IJCAI)},
  pages = {4127--4135},
  publisher = {International Joint Conferences on Artificial Intelligence Organization},
  year = {2023},
  doi = {10.24963/ijcai.2023/459},
  eprint = {2305.10292},
  archivePrefix = {arXiv},
  url = {https://www.ijcai.org/proceedings/2023/459}
}

@article{Pham2025,
  author = {Pham, Canh V.},
  title = {Enhanced Deterministic Approximation Algorithm for Non-monotone Submodular Maximization under Knapsack Constraint with Linear Query Complexity},
  journal = {Journal of Combinatorial Optimization},
  volume = {49},
  number = {1},
  pages = {2},
  year = {2025},
  doi = {10.1007/s10878-024-01232-9},
  url = {https://doi.org/10.1007/s10878-024-01232-9},
  note = {Article 2}
}

@misc{LiuZ2026,
  author = {Liu, Zihui and Zhang, Zhijie},
  title = {Linear-Query Deterministic Approximation for Non-monotone Submodular Maximization under a Knapsack Constraint},
  year = {2026},
  eprint = {2609.25679},
  archivePrefix = {arXiv},
  primaryClass = {cs.DS},
  doi = {10.48550/arXiv.2609.25679},
  url = {https://arxiv.org/abs/2609.25679},
  howpublished = {arXiv:2609.25679}
}

@inproceedings{feldman2011,
  author    = {Moran Feldman and Joseph Naor and Roy Schwartz},
  title     = {A Unified Continuous Greedy Algorithm for Submodular Maximization},
  booktitle = {Proceedings of the 52nd Annual IEEE Symposium on Foundations of Computer Science (FOCS)},
  pages     = {570--579},
  year      = {2011},
  doi       = {10.1109/FOCS.2011.46}
}

@inproceedings{ene2016,
  author    = {Alina Ene and Huy L. Nguyen},
  title     = {Constrained Submodular Maximization: Beyond $1/e$},
  booktitle = {Proceedings of the 57th Annual IEEE Symposium on Foundations of Computer Science (FOCS)},
  pages     = {248--257},
  year      = {2016},
  doi       = {10.1109/FOCS.2016.34}
}

@article{buchbinder2019,
  author  = {Niv Buchbinder and Moran Feldman},
  title   = {Constrained Submodular Maximization via a Nonsymmetric Technique},
  journal = {Mathematics of Operations Research},
  volume  = {44},
  number  = {3},
  pages   = {988--1005},
  year    = {2019},
  doi     = {10.1287/moor.2018.0955}
}

@article{filmus2014,
  author  = {Yuval Filmus and Justin Ward},
  title   = {Monotone Submodular Maximization over a Matroid via Non-Oblivious Local Search},
  journal = {SIAM Journal on Computing},
  volume  = {43},
  number  = {2},
  pages   = {514--542},
  year    = {2014},
  doi     = {10.1137/130920277}
}

@misc{Aggarwal2026Geometry,
      title={Geometry-Dependent Bounds for Online Non-Monotone {DR}-Submodular Maximization}, 
      author={Vaneet Aggarwal},
      year={2026},
      eprint={2610.00545},
      archivePrefix={arXiv},
      primaryClass={cs.LG},
      url={https://arxiv.org/abs/2610.00545}, 
      howpublished = {arXiv:2610.00545}
}
